\documentclass[format=acmsmall, review=false]{acmart}
\pdfoutput=1
\usepackage{acm-ec-18}

\usepackage{times}
\usepackage{soul}
\usepackage{url}
\usepackage[utf8]{inputenc}
\usepackage{caption}
\usepackage{graphicx}
\usepackage{amsmath}
\usepackage{amsthm}
\usepackage{booktabs}
\usepackage{amsmath,amsthm,amssymb,algorithm,algpseudocode,graphicx,tikz,caption,subcaption,xcolor,bbm}
\definecolor{light-gray}{gray}{0.7}
\usepackage{amssymb}% http://ctan.org/pkg/amssymb
\usepackage{pifont}% http://ctan.org/pkg/pifont
\newcommand{\cmark}{\ding{51}}%
\newcommand{\xmark}{\ding{55}}%

\usepackage{adjustbox}
\usepackage{pifont}
\usepackage{wrapfig}
\usepackage{adjustbox}
\usepackage{graphicx}
\usepackage{xcolor,makecell}
\usepackage{booktabs}
\usepackage{multirow}
\usepackage{stackengine}
\usepackage{csquotes}
\usepackage{natbib}
\numberwithin{equation}{section}
\allowdisplaybreaks

\usepackage{hyperref}

\stackMath

\makeatletter
\def\moverlay{\mathpalette\mov@rlay}
\def\mov@rlay#1#2{\leavevmode\vtop{%
		\baselineskip\z@skip \lineskiplimit-\maxdimen
		\ialign{\hfil$\m@th#1##$\hfil\cr#2\crcr}}}
\newcommand{\charfusion}[3][\mathord]{
	#1{\ifx#1\mathop\vphantom{#2}\fi
		\mathpalette\mov@rlay{#2\cr#3}
	}
	\ifx#1\mathop\expandafter\displaylimits\fi}
\makeatother

\newcommand{\cupdott}{\charfusion[\mathbin]{\cup}{\cdot}}
\newcommand{\bigcupdott}{\charfusion[\mathop]{\bigcup}{\cdot}}

\algrenewcommand\algorithmicindent{0.43em}
\algnotext{EndFor}
\algnotext{EndIf}
\algnotext{EndWhile}
\algnotext{EndProcedure}

\makeatletter
\providecommand*{\cupdot}{%
	\mathbin{%
		\mathpalette\@cupdot{}%
	}%
}
\newcommand*{\@cupdot}[2]{%
	\ooalign{%
		$\m@th#1\cup$\cr
		\hidewidth$\m@th#1\cdot$\hidewidth
	}%
}
\def\NoNumber#1{{\def\alglinenumber##1{}\State #1}\addtocounter{ALG@line}{-1}}
\allowdisplaybreaks
\begin{document}

\title{Stable Roommate Problem with Diversity Preferences}
\titlenote{A shortened version of this paper appears in the proceedings of 
IJCAI'20 and an extended  abstract in the proceedings of AAMAS'20.}
\author{Niclas Boehmer}
\affiliation{TU Berlin, Berlin, Germany}
\author{Edith Elkind}
\affiliation{University of Oxford, Oxford, U.K.}

\begin{abstract}
  In the multidimensional stable roommate problem, agents have to be allocated to rooms
  and have preferences over sets of potential roommates. We study
  the complexity of finding good allocations of agents to rooms under the assumption
  that agents have diversity preferences \citep{bredereck2019hedonic}: each agent
  belongs to one of the two types (e.g., juniors and seniors, artists and engineers),
  and agents' preferences over rooms depend solely on the fraction of 
  agents of their own type among their potential roommates. We consider various solution 
  concepts for this setting, such as core and exchange stability, Pareto optimality 
  and envy-freeness.
  On the negative side, we prove that envy-free, core stable or (strongly) exchange stable outcomes 
  may fail to exist and that the associated decision problems 
  are NP-complete. On the positive side, we show that these problems are in FPT
  with respect to the room size, which is not the case for the general stable 
  roommate problem. Moreover, for the classic
  setting with rooms of size two, we present a linear-time algorithm that 
  computes an outcome 
  that is core and exchange stable as well as Pareto optimal. 
  Many of our results for the stable roommate problem extend to the stable 
  marriage problem.
  \vspace{-10pt}
\end{abstract}

\maketitle

\section{Introduction} \label{ro::MaR}
Alice and Bob are planning their wedding. They have agreed on the gift registry 
and the music to be played, but they still need to decide on the 
seating plan for the wedding reception. They expect 120 guests, and the reception venue
has 20 tables, with each table seating 6 guests. However, this task is far from being easy:
e.g., Alice's great-aunt does not get along with Bob's family and prefers
not to share the table with any of them; on the contrary, Bob's younger brother is keen
to meet Alice's family and would be upset if he were stuck with his relatives.
After spending an evening trying to find a seating plan
that would keep everyone happy, Alice and Bob are on the brink of canceling the wedding 
altogether.

Bob's friend Charlie wonders if the hapless couple may benefit from consulting 
the literature on the {\em stable roommate problem}. In this problem, 
the goal is to find a stable assignment of $2n$ agents 
into $n$ rooms of size $2$, where every agent has a preference relation over her possible 
roommates. The most popular notion of stability in this context is {\em core stability}:
no two agents should strictly prefer each other to their current roommate. Another
relevant notion is {\em exchange stability}: no two agents should want to swap their places.
However, for the stable roommate problem, neither core
stable nor exchange stable outcomes are guaranteed to exist.
Further, while \citet{irving1985efficient} proved that it is 
possible to decide in time linear in the size of the input if an instance of 
the roommate problem 
with strict preferences admits a core stable outcome, many other algorithmic 
problems
for core and exchange stability are computationally hard
\citep{ronn1990np,cechlarova2005exchange}. 
For the {\em $s$-dimensional stable roommate problem},
where each room has size $s$ and agents have preferences over all $s-1$-subsets 
of 
agents as their potential roommates, even the core non-emptiness problem for 
strict preferences
is NP-complete for $s\ge 3$ \citep{ng1991three,huang2007two}.

However, Charlie then notes that Alice and Bob's problem has additional structure:
the invitees can be classified as bride's family or groom's family, and it appears
that all constraints on seating arrangements can be expressed in terms of this 
classification: each person only has preferences over the ratio of groom's relatives
and bride's relatives at her table. Thus, the problem in question is closely related
to {\em hedonic diversity games}, recently introduced by \citet{bredereck2019hedonic}.
These are coalition formation games where agents have {\em diversity preferences}, 
i.e., they are partitioned into two 
groups (say, \textit{red} and \textit{blue}), and every agent is indifferent among 
coalitions with the same ratio of red and blue agents. However, positive results
for hedonic diversity games are not directly applicable to the roommate setting: 
in hedonic games, agents
form groups of varying sizes, while the wedding guests have to be split into groups of six.

In this paper, we investigate 
the multidimensional stable roommate problem (for arbitrary room size $s$) with diversity preferences;  
we refer to the resulting problem as the \textit{roommate diversity problem}. 
This model captures important aspects of several real-world group formation scenarios, 
such as flat-sharing, splitting students into teams for group projects, and 
seating arrangements at important events. We consider common solution
concepts from the literature on the stable roommate problem; for each
solution concept, we analyze the complexity of checking if a given outcome 
is a valid solution, whether the set of solutions is guaranteed to be non-empty, 
and, if not, how hard it is to check if an instance admits a solution
as well as to compute a solution if it exists.

%%%%%%%%%%%%%%%%%%%%%%%%%%%%%%%%%%%%%%%%%%%%%%%%%%%%%%%%%%%%%%%%%%%%%%%
\begin{table*}[t!]\centering
		\begin{adjustbox}{max width=\textwidth}
		\begin{tabular}{@{}clllllllllllll@{}}
			\toprule
			& 
			\multicolumn{3}{c}{unrestricted} & 
			& 
			\multicolumn{3}{c}{strict} & 
			& 
			\multicolumn{3}{c}{dichotomous}
			\\
			\cmidrule{2-4} 
			\cmidrule{6-8} 
			\cmidrule{10-12}
			& Gu. & \multicolumn{1}{c}{Ex.} & \multicolumn{1}{c}{Co.} && Gu. & 
			\multicolumn{1}{c}{Ex.} & \multicolumn{1}{c}{Co.} && 
			Gu. & \multicolumn{1}{c}{Ex.} & \multicolumn{1}{c}{Co.}
			\\
			\midrule
			% Strong core & \hyperref[th::Ro_StCor_Ex]{\xmark}    
			%& \hyperref[th::Roo_SC_Npc]{NP-c.}  &&  
			%\hyperref[th::Ro_StCor_Ex]{\xmark} & \textbf{?} && 
			% \hyperref[th::Ro_StCor_Ex]{\xmark}  & \textbf{?} 
			%\\
			Core &  \hyperref[th:exCore]{\xmark} (\ref{th:exCore}) & 
			\hyperref[th::Roo_Core_NPc]{NPc} (\ref{th::Roo_Core_NPc}) & 
			\hyperref[th::Roo_Core_NPc]{NPh} (\ref{th::Roo_Core_NPc})&& 
			\hyperref[th:exCore]{\xmark} (\ref{th:exCore})& 
			\hyperref[th::Roo_Core_NPc]{NPc} (\ref{th::Roo_Core_NPc}) & 
			\hyperref[th::Roo_Core_NPc]{NPh} (\ref{th::Roo_Core_NPc}) && 
			\hyperref[th:exCoreSP]{\cmark} (\ref{th::RoomDich})& - & P 
			(\ref{th::RoomDich})
			\\
			Strong Core &  \hyperref[th:exCore]{\xmark} (\ref{th:exCore}) & 
			\hyperref[th::Roo_SC_Npc]{NPc} (\ref{th::Roo_SC_Npc}) & 
			\hyperref[th::Roo_SC_Npc]{NPh} (\ref{th::Roo_SC_Npc})&& 
			\hyperref[th:exCore]{\xmark} (\ref{th:exCore})& 
			\hyperref[th::Roo_SC_Npc]{NPc} (\ref{th::Roo_SC_Npc}) & 
			\hyperref[th::Roo_SC_Npc]{NPh} (\ref{th::Roo_SC_Npc}) && 
			\hyperref[th::Roo_SC_Npc]{\xmark} (\ref{th::Roo_SC_Npc})& 
			\hyperref[th::Roo_SC_Npc]{NPc} (\ref{th::Roo_SC_Npc}) & 
			\hyperref[th::Roo_SC_Npc]{NPh} (\ref{th::Roo_SC_Npc})
			\\
			Same Ex. &  	\hyperref[th::Room_SSwap]{\cmark} 
			(\ref{th::Room_SSwap}) & - & \hyperref[th::Room_SSwap]{P} 
			(\ref{th::Room_SSwap}) && 
			\hyperref[th::Room_SSwap]{\cmark} (\ref{th::Room_SSwap}) & - & 
			\hyperref[th::Room_SSwap]{P} (\ref{th::Room_SSwap}) && 
			\hyperref[th::Room_SSwap]{\cmark} (\ref{th::Room_SSwap}) & - & 
			\hyperref[th::Room_SSwap]{P} (\ref{th::Room_SSwap}) 
			\\
			Exch. &  	\hyperref[th:SSSP]{\xmark} (\ref{th:SSSP}) & \textbf{?} 
			& \textbf{?} &&  \hyperref[th:SSSP]{\xmark} (\ref{th:SSSP}) & 
			\textbf{?} & 
			\textbf{?} && 
			\textbf{?} & \textbf{?} & \textbf{?} \\
			Strong Ex. & \hyperref[th:SSSP]{\xmark} (\ref{th:SSSP}) & 
			\hyperref[th:SSS]{NPc} (\ref{th:SSS}) & \hyperref[th:SSS]{NPh} 
			(\ref{th:SSS}) &&
			\hyperref[th:SSSP]{\xmark} (\ref{th:SSSP}) & \textbf{?} & 
			\textbf{?} &&
			\hyperref[th:SSSP]{\xmark} (\ref{th:SSSP}) & \hyperref[th:SSS]{NPc} 
			(\ref{th:SSS}) & \hyperref[th:SSS]{NPh} (\ref{th:SSS})
			\\
			Pareto & \cmark  & - & \hyperref[th:Rea]{NPh} (\ref{th:Rea}) && 
			\cmark  & - & \textbf{?} && 
			\cmark  & - & \hyperref[th:Rea]{NPh} (\ref{th:Rea})
			\\
			Envy-free &  	\hyperref[th::Room_Envy]{\xmark} 
			(\ref{th::Room_Envy}) & \hyperref[th::Room_Envy_NPc]{NPc 
			(\ref{th::Room_Envy_NPc})}& \hyperref[th::Room_Envy_NPc]{NPh 
			(\ref{th::Room_Envy_NPc})} && 
			\hyperref[th::Room_Envy]{\xmark} (\ref{th::Room_Envy}) & 
			\hyperref[th::Room_Envy_NPc]{NPc (\ref{th::Room_Envy_NPc})} & 
			\hyperref[th::Room_Envy_NPc]{NPh (\ref{th::Room_Envy_NPc})} && 
			\hyperref[th::Room_Envy]{\xmark} (\ref{th::Room_Envy}) & 
			\hyperref[th::Room_Envy_NPc]{NPc (\ref{th::Room_Envy_NPc})} & 
			\hyperref[th::Room_Envy_NPc]{NPh (\ref{th::Room_Envy_NPc})}
			\\
			\bottomrule
		\end{tabular}
	\end{adjustbox}
	\caption{Overview of complexity results for different solution concepts 
		and restrictions on the preference relations. For each solution concept 
		and restriction, we indicate whether every instance 
		satisfying this restriction is guaranteed to admit an outcome with 
                the respective property (Gu.), the complexity of deciding if an 
                instance admits 
                such an outcome (Ex.) and of finding one if it exists (Co.).
		The number in brackets is the number of the respective theorem.
		For all solution concepts, the problem of verifying whether a given 
		outcome
		has the desired property is in P except for Pareto optimality,
		for which this problem is coNP-complete.
		We prove that all existence problems in this table are in FPT with 
		respect to the room size.}
	\label{t:Sum}
\end{table*}

\smallskip

\noindent\textbf{\textit{Our Contribution}\ } 
%In this paper, we lay the strategic foundations of the algorithmic analysis of the roommate 
%diversity problem. 
%
To begin with, we show that for room size two, every roommate diversity problem 
admits an outcome 
that is core stable, exchange stable and Pareto optimal; moreover, we show that 
such an outcome can be computed in 
linear time. For $s>2$, we provide
counterexamples showing that 
core stable, exchange stable or envy-free outcomes may fail to exist; 
these negative results hold even when agents' preferences over type ratios
are restricted to be strict and single-peaked (see Section~\ref{sec:prelim}
for formal definitions). We also prove that for core stability, strong exchange
stability and envy-freeness, the existence questions are computationally hard;
for Pareto optimality, we show that it is not only hard to find a Pareto optimal outcome, 
but also to verify whether a given outcome is Pareto optimal. 
We provide an overview of our results in Table \ref{t:Sum}. 
%For all studied properties, the problem of verifying whether an outcome fulfills this 
%property is in P except for Pareto optimality where this problem is coNP-complete.

On the positive side, we show that all existence questions 
we consider are in FPT with respect to the room size.
To the best of our knowledge, apart from some work on the multidimensional 
stable roommate problem with cyclic preferences \citep{hofbauer2016ddimensional}, 
this is one of the first positive results for
the multidimensional stable roommate problem. Thus, the roommate diversity
problem offers an attractive combination of expressive power 
and computational tractability.

In the end, we also investigate 
the multidimensional stable marriage problem with diversity preferences.
In our model, there are $s$ agent categories, with $n$ agents in each category, 
and each agent can be classified as red or blue irrespective of her category.
An admissible outcome is a partition of agents into $n$ groups, 
with each group having one representative from each category; the agents' 
preferences 
over groups are governed by the fraction of red and blue agents in each group. 
We 
show that 
our negative results for the roommate diversity problem 
extend to this model, but for the positive results the picture is more 
complicated.

%%%%%%%%%%%%%%%%%%%%%%%%%%%%%%%%%%%%%%%%%%%%%%%%%%%%%%%%%%%%%%%%%%%%%%%%%%%%%%%%%%
\smallskip

\noindent\textbf{\textit{Related Work}\ } 
The stable roommate problem was proposed by \citet{gale1962college} and has 
been studied 
extensively since then 
\cite{cechlarova2002complexity,irving1985efficient,irving2002stable,ronn1990np,huang2007two}.
It can be seen as a special case of {\em hedonic coalition formation} 
\cite{bogomolnaia2002stability}, where agents have to split into groups 
(with no prior constraints on the group sizes) and have preferences over groups
that they can be part of; precisely, a stable roommate problem with room size 
$s$ 
is a hedonic coalition formation problem where each agent considers all 
coalitions 
of size other than $s$ unacceptable.

In contrast to the closely related stable marriage problem 
\cite{gale1962college}, 
for the stable roommate problem, it is not guaranteed that the core is not 
empty. While 
\citet{irving1985efficient} proved that it is possible to check in time linear 
in the input whether a roommate problem admits a core stable outcome if the 
preferences are 
restricted to be strict, \citet{ronn1990np} showed that this problem
becomes NP-complete if ties in the preference relations are allowed.
As in practice a group deviation usually requires some regrouping, 
\citet{alcalde1994exchange} initiated the study of local stability notions 
that do not require reallocating non-deviating agents, by introducing the notion of exchange 
stability.  Subsequently, \citet{cechlarova2005exchange} 
proved that it is NP-complete to 
decide whether an instance of the roommate problem with strict preferences 
admits an exchange stable outcome.
Another concept that is relevant for the roommate problem is 
Pareto optimality 
\citep{abraham2004pareto,sotomayor2011pareto,DBLP:journals/corr/abs-1901-06737}.
\citet{morrill2010roommates} proved that for room size two, 
it is possible to check if a given outcome
is Pareto optimal, and to find a Pareto improvement if it exists, 
in time cubic in the number of agents. This implies that 
a Pareto optimal outcome can be found in polynomial time.
Researchers have also considered various notions of fairness 
in the context of the roommate problem 
\citep{aziz2019random,abdulkadirouglu2003school}. One 
such notion is envy-freeness: an outcome is said to be {\em envy-free} 
if no agent wants to take the place of another agent.

While much of the work on the stable roommate problem focuses on the case $s=2$,
there are a few papers that consider the three-dimensional stable roommate problem
\citep{ng1991three, huang2007two}. In some of this work, the agents' 
preferences are defined over individual agents 
and then lifted to pairs \cite{iwama2007stable, 
	huang2007two}, while other authors define preferences directly over pairs 
\cite{ng1991three, huang2007two}. In both models, it 
is NP-complete to decide whether an instance with strict preferences admits a core stable 
outcome \citep{iwama2007stable,ng1991three,huang2007two}. For this reason, despite its great 
practical relevance, the multidimensional version of the stable roommate 
problem 
has not attracted much attention yet.

One possibility to circumvent these negative results is to search for reasonable subclasses 
of the stable roommate problem, i.e., to identify realistic restrictions on 
the agents' preferences for which the associated computational problems become tractable.
This approach has been successful in the study of the two-dimensional stable roommate problem
\citep{cseh2019pairwise,abraham2007stable,bartholdi1986stable,bredereck2017stable,chung2000existence},
as well as in the context of hedonic games \citep{banerjee2001core,bogomolnaia2002stability,aziz2019fractional}.
In particular, \citet{bredereck2019hedonic} and \citet{be20} 
analyzed the complexity of finding stable outcomes
in hedonic diversity games for several notions of stability, such
as Nash stability, individual stability and core stability. However, 
these results do not directly translate to our model: first, 
with the exception of core stability, the solution concepts we consider
are different from those considered in hedonic diversity games, 
and second, the hard constraint on the room sizes in the 
roommate problem changes both the set of feasible allocations and the set
of possible deviations.  

Finally, we note that \citet{bredereck2019hedonic} and \citet{be20}
related hedonic games with diversity preferences to {\em anonymous hedonic games}
\citep{bogomolnaia2002stability}, where the agents' preferences over 
coalitions only depend on coalition sizes. This connection proves useful 
in our setting as well: 
%since in the roommate diversity problem 
%we can describe each room by specifying the number of red agents in it.
%This enables us to 
%use anonymous hedonic games in our algorithmic analysis 
%of the roommate diversity problem. In particular, 
e.g., in our hardness
reductions we use the fact that it is NP-complete 
to decide whether an anonymous hedonic game admits a Nash, core or individually stable 
outcome~\citep{ballester2004np}.

Diversity related questions have been also previously studied in the context of 
stable matching problems \citep{huang2010classified,kamada2015efficient}. 
However, our model is 
fundamentally different, as, 
here, agents have preferences over types and, in contrast to previous research, 
types are not used to formulate distributional constraints on the outcome. 

%%%%%%%%%%%%%%%%%%%%%%%%%%%%%%%%%%%%%%%%%%%%%%%%%%%%%%%%%%%%%%%%%%%%%%%%%%%%%

\section{Preliminaries}\label{sec:prelim}
For every positive integer $t$, we denote the set $\{1, \dots, t\}$ by $[t]$, 
and we write $[0, t]$ to denote $\{0\}\cup[t]$.
\begin{definition}\label{def:rdp}
	A {\em roommate diversity problem} with room size $s$ is a quadruple 
	$G=(R,B,s,$ $ (\succsim_{i})_{i\in R\cup B})$ with $N=R\cupdott B$ and 
	$|N|=k\cdot s$ 
	for some $k\in \mathbb{N}$. The preference relation $\succsim_{i}$ 
	of each agent $i\in N$ is a weak order over the set $D=\{\frac{j}{s} : j\in [0, s]\}$.
\end{definition}
In the following, we call all agents in $R$ {\em red} agents 
and all agents in $B$ {\em blue} agents, and write
$r=|R|$, $b=|B|$ and $n=|R\cupdott B|$ . 
We refer to size-$s$ subsets of $N$ as {\em coalitions} or {\em rooms};
the quantity $k=\frac{|N|}{s}$ is then the number of rooms. 
For each $i\in N$, let $\mathcal{N}_{i}=\{S\subseteq N: |S|=s, i\in S\}$ 
denote all size-$s$ subsets of $N$ containing $i$, i.e., all possible rooms that 
$i$ can be part of. 

An {\em outcome} of $G$ is a partition of all agents into $k$ rooms $\pi=\{C_1,\dots,C_k\}$ 
such that $|C_i|=s$ for all $i\in [k]$. 
Let $\pi(i)$ denote $i$'s coalition in $\pi$. Given a coalition $C\subseteq N$, 
let $\theta(C)$ denote the fraction of red agents in $C$, i.e., 
$\theta(C)=\frac{|C\cap R|}{|C|}$: we say that $C$ {\em is of fraction $\theta(C)$}. 
A coalition $C\subseteq N$ is called {\em pure} if $C\subseteq R$ or $C\subseteq B$ 
and {\em mixed} otherwise. 

For each agent $i\in N$, we interpret her preference relation $\succsim_i$ over $D$ as her 
preferences over the fraction of red agents in her coalition; for instance, 
$\frac{2}{5}\succsim_{i}\frac{3}{5}$ means that agent $i$ prefers a room 
where two out of five agents 
are red to a room where three out of five agents are red.\footnote{We could equivalently 
	define the agents' preferences over the number of red agents in each room; we chose the ratio-based
	definition for consistency with the hedonic diversity games literature 
	and to emphasize the room size.} 
Thereby, $\succsim_i$ induces agent $i$'s preferences over all possible rooms she can be 
part of.\footnote{For succinctness and consistency with prior work, we assume that each agent has preferences 
over the entire set $D$, including $0$ and $1$, even though a red agent cannot 
be part of a room with ratio $0$ and a blue agent cannot be part of a room with ratio $1$. 
Allowing agents to have preferences over `impossible' ratios has no impact on our results:
even if, say, a blue agent ranks $1$ highly, she cannot deviate 
to a coalition with this ratio. In all of our examples, the impossible 
ratios are ranked at the bottom of agent's preferences.} 
Given two rooms $S,T\in \mathcal{N}_i$, overloading notation, we write $S\succ_i T$ and say that $i$ 
{\em strictly prefers} $S$ to $T$ if $\theta(S)\succsim_{i}\theta(T)$ and 
$\theta(T)\not\succsim_{i}\theta(S)$. Further,  
we write $S\succsim_i T$ and say that $i$ 
{\em weakly prefers} $S$ to $T$ if $\theta(T)=\theta(S)$ or 
$\theta(T)\succsim_{i}\theta(S)$. If $i$ weakly prefers $S$ to $T$ and $T$ to $S$, 
we write $S\sim_i T$ and say that $i$ is {\em indifferent} between $S$ and $T$.

The preference relation of agent $i$ is said to be {\em single-peaked} 
if there exists a peak $p_i\in D$ such that 
for all $\alpha, \beta \in D$ such that  
$p_i\leq \alpha < \beta$ or $\beta <\alpha \leq p_i$ it holds that 
$\alpha \succsim_i \beta$. The preference relation of $i$ 
is said to be {\em dichotomous} if it is possible to partition $D$ 
into two sets $D^+$ and $D^-$ so that 
for all $d\in D^+, d'\in D^-$ it holds that $d\succ_i d'$, 
for all $d,d'\in D^-$ it holds that $d\sim_i d'$ and 
for all $d,d'\in D^+$ it holds that $d\sim_i d'$. 
We say that $i$ {\em approves} all fractions in $D^+$ and {\em disapproves} all fractions in $D^-$. 

Generalizing the definition of \citet{gale1962college}, 
we say that a coalition $S\subseteq N$ with $|S|=s$ {\em blocks} an outcome $\pi$ if 
for all $i\in S$ it holds that $\theta(S)\succ_{i}\theta(\pi(i))$; 
we say that $S$ {\em weakly blocks} an outcome $\pi$ if 
for all         $i\in S$ it holds that $\theta(S)\succsim_{i}\theta(\pi(i))$ and 
there exists an $i\in S$ such that $\theta(S)\succ_{i}\theta(\pi(i))$. 
An outcome $\pi$ is called {\em (strongly) core stable} if no coalition (weakly) blocks it.

In an outcome $\pi$, a pair of agents $i,j\in N$ with $\pi(i)\neq \pi(j)$ has an 
{\em exchange deviation} if they would like to exchange places, i.e., 
$\theta\big((\pi(j)\setminus \{j\})\cup \{i\} \big) \succ_{i} \theta\big(\pi(i)\big)$ and 
$\theta\big((\pi(i)\setminus \{i\})\cup \{j\}\big)  \succ_{j} \theta\big(\pi(j)\big)$. 
Further, a pair of agents $i,j\in N$ with 
$\pi(i)\neq \pi(j)$ has a {\em weak exchange deviation} if 
$\theta\big((\pi(j)\setminus \{j\})\cup \{i\}\big)\succ_{i}\theta\big(\pi(i)\big)$ and 
$\theta\big((\pi(i)\setminus \{i\})\cup \{j\}\big)\succsim_{j}\theta\big(\pi(j)\big)$. 
An outcome is called {\em (strongly) exchange stable} 
if no pair of agents has a (weak) exchange deviation~\citep{alcalde1994exchange}.

An outcome $\pi$ is called {\em Pareto optimal} 
if there is no other outcome that makes all agents weakly better off and some agents 
strictly better off, i.e., 
there is no outcome $\pi'$ such that 
$\theta(\pi'(i))\succsim_{i} \theta(\pi(i))$ for all $i\in N$ and 
$\theta(\pi'(i))\succ_{i} \theta(\pi(i))$ for some $i\in N$.

An outcome $\pi$ is called {\em envy-free} if there does not exist a pair of agents $i,j\in N$ 
with $\pi(i)\neq \pi(j)$ such that $i$ envies $j$'s place, i.e., 
$\theta\big((\pi(j)\setminus \{j\}) \cup \{i\} \big) \succ_{i} \theta\big(\pi(i)\big)$.

Besides the study of the classic computational complexity of a problem, 
we can also investigate its complexity as a function of a specific parameter 
$\rho$
of the input; e.g., in the context of the roommate diversity problem, this 
parameter could
be the size of the rooms $s$ or the size of the smaller of the two groups 
$\min\{r, b\}$.
We say that a problem is {\em fixed-parameter tractable with respect to a 
parameter}
if, given an instance $I$ where this parameter takes value $\rho$, 
it is solvable in time 
$f(\rho)|I|^{\mathcal{O}(1)}$, where $f$ is some computable 
function.\footnote{For 
	further details on the theory of parameterized complexity and the 
	significance of $\mathit{FPT}$, the reader is referred to textbooks by 
	\citet{downey2013fundamentals} and by \citet{niedermeier2006invitation}.}

%%%%%%%%%%%%%%%%%%%%%%%%%%%%%%%%%%%%%%%%%%%%%%%%%%%%%%%%%%%%%%%%%%%%%%%%%%%%%%%%%%%%%

\section{Roommate Diversity Problem With Room Size Two} \label{se::Two}
For $s=2$, the roommate diversity problem becomes a special case of the classic stable 
roommate problem. Our first observation is that diversity preferences make the classic 
problem significantly easier: we can efficiently find an outcome that is Pareto optimal, 
core stable and exchange stable. This result motivates us to focus on the case $s>2$
in the remainder of the paper.

\begin{theorem} \label{th::Room2}
	Every instance of the roommate diversity problem with room size two admits an outcome that 
	is Pareto optimal, core stable and exchange stable, which can be computed 
	in linear time, even if we allow for 
	indifferences in the preferences.
\end{theorem}
\begin{proof}
	In the following, we prove the theorem by proving that every outcome 
	returned by Algorithm~\ref{al:n2}, which runs obviously in linear time, is 
	guaranteed to be core stable, exchange stable and Pareto optimal. However, 
	we start by giving a high-level description of the algorithm. 
	We say that a mixed pair is {\em happy} if both agents in the pair weakly
	prefer a mixed pair to a pure pair.
	Algorithm~\ref{al:n2} first create as many happy mixed pairs as possible.
	The algorithm then attempts to put the remaining agents in pure pairs. 
	If at this point the number of blue agents 
	is odd, this is not possible. 
	In this case, depending on the preferences of agents in mixed pairs, 
	it either inserts an additional mixed pair or breaks up one of the mixed pairs 
	to create two pure pairs: one red and one blue. 
	
	\begin{algorithm}[t] 
		\caption{Computing a stable outcome in an instance of the 2-dimensional 
			roommate 
			diversity problem} \label{al:n2}
		\begin{algorithmic}[1]
			\Require{An instance $G=(R,B,2, (\succ_{i})_{i\in N})$} 
			\Ensure{Pareto optimal, core and exchange stable outcome $\pi$ of 
				$G$}
			\State {Let $B^{*}=\{b^{*}_{1},\dots,b^{*}_{k}\}$ 
				be the set of all blue agents $b$ such that 
				$\frac{1}{2}\succsim_{b}0$, 
				starting with all agents for which this relation is strict.}
			\State{Let $R^{*}=\{r^{*}_{1},\dots, r^{*}_{\ell}\}$ 
				be the set of all red agents $r$ such that 
				$\frac{1}{2}\succsim_{r}1$, 
				starting with all agents for which this relation is strict.}  
			\State {Let $B^{'}=\{b^{'}_{1},\dots,b^{'}_{p}\}=
				B\setminus\{b^{*}_{1},\dots,b^{*}_{\min\{k,\ell\}}\}$,  
				$R^{'}=\{r^{'}_{1},\dots,r^{'}_{q}\}=
				R\setminus\{r^{*}_{1},\dots,r^{*}_{\min\{k,\ell\}}\}$, 
				starting with all agents who strictly prefer a pure coalition 
				followed by all agents who are indifferent 
				followed by all agents who strictly prefer a mixed coalition. 
			}      
			\If{$p\mod 2=0$}
			\State \Return 
			$\pi=\{\{b^{*}_{i},r^{*}_{i}\}: i\in [\min\{k,\ell\}]\}\cup 
			\{\{b^{'}_{i},b^{'}_{i+1}\}: i\in \{1,3,\dots,p-3,p-1\}\} 
			\cup$ 
			\NoNumber{\hspace{1.6cm}
				$\{\{r^{'}_{i},r^{'}_{i+1}\}: i\in \{1,3,\dots,q-3,q-1\}\}$}
			\Else 
			\If{$0\succsim_{b^{*}_{\min\{k,\ell\}}}\frac{1}{2}$ {\bf and}
				$1\succsim_{r^{*}_{\min\{k,\ell\}}}\frac{1}{2}$}
			\State \Return $\pi=\{\{b^{*}_{i},r^{*}_{i}\}: i\in 
			[\min\{k,\ell\}-1]\}\cup\{\{b^{'}_{i},b^{'}_{i+1}\}: i\in 
			\{1,3,\dots,p-4, p-2\}\}\cup$
			\NoNumber{\hspace{1.7cm}$\{\{r^{'}_{i},r^{'}_{i+1}\}: 
				i\in \{1,3,\dots,q-4,q-2\}\}\cup 
				\{\{r^{'}_{q},r^{*}_{\min\{k,\ell\}}\},  
				\{b^{'}_{p},b^{*}_{\min\{k,\ell\}}\}\}$}
			\Else
			\State \Return $\pi=\{\{b^{*}_{i},r^{*}_{i}\}: i\in 
			[\min\{k,\ell\}]\}\cup\{\{b^{'}_{i},b^{'}_{i+1}\}: 
			i\in \{1,3,\dots,p-4,p-2\}\}\cup$
			\NoNumber{\hspace{1.7cm}$ \{\{r^{'}_{i},r^{'}_{i+1}\}: 
				i\in \{1,3,\dots,q-4,q-2\}\} \cup \{\{r^{'}_{q},b^{'}_{p}\}\}$}
			\EndIf
			\EndIf
		\end{algorithmic}
	\end{algorithm}

	\noindent \textbf{\textit{Core stability}\ } First of all, note that the 
	outcome $\pi$ computed by Algorithm 1 is core stable, as for each agent in 
	a pure coalition it either holds that this is one of her most preferred 
	coalitions or all agents of the other type strictly preferring mixed 
	coalitions are already in a mixed coalition. Moreover, for agents in a 
	mixed coalition it either holds that this is one of their most preferred 
	coalitions or all agents of the same type strictly preferring a pure 
	coalition are already in one.

	\noindent \textbf{\textit{Exchange stability}\ } To prove that $\pi$ is 
	exchange stable, we examine all agents who might have an incentive to swap 
	and prove that there never exists an agent willing to swap with them: an 
	agent in one of the coalitions from $\{\{b^{*}_{i},r^{*}_{i}\} : i\in 
	[\min(\ell,k)]\}$ never strictly prefers swapping with another agent, as 
	she is already in one of her most preferred coalitions. 
	
	If 
	$\{r^{'}_{q},b^{'}_{p}\}\in \pi$, only in the case that~$1\succ_{r^{'}_{q}} 
	\frac{1}{2}$, agent $r'_{q}$ strictly prefers to swap with a red agent in a 
	pure coalition or with a blue agent in a mixed coalition. However, as 
	argued above, no blue agent in one of the other mixed coalitions has an 
	incentive to swap. Moreover, by construction,  $1\succ_{r^{'}_{q}} 
	\frac{1}{2}$ implies that $1\succ_{r^{'}} \frac{1}{2}$ for all red agents 
	$r^{'}$ in pure coalitions. Thereby, no red agent in a pure coalition wants 
	to swap with $r^{'}_{q}$. The same holds if $b^{'}_{p}$ has an incentive to 
	swap. 
	
	Lastly, a red agent $r$ in a pure coalition with 
	$\frac{1}{2}\succ_{r} 1$ strictly prefers swapping with a red agent in a 
	mixed coalition or a blue agent in a pure coalition. However, from the fact 
	that a red agent with $\frac{1}{2}\succ_{r} 1$ is in a pure coalition it 
	follows that all red agents in mixed coalitions also have this preference 
	relation. Moreover, it also follows that all blue agents who weakly prefer 
	a mixed coalition are already in a mixed coalition. For blue agents in pure 
	coalitions with an incentive to swap, the reasoning is analogous. Note that 
	if the preferences of agents are strict, $\pi$ is also guaranteed to be 
	strongly exchange stable.

	\noindent \textbf{\textit{Pareto optimality}\ }
	For the sake of contradiction, let us assume that there exists an outcome 
	$\pi'$ that weakly Pareto dominates $\pi$. By definition, there needs to 
	exist at least one agent $i\in N$ who strictly prefers $\pi'$ to $\pi$. 
	Consequently, $i$ cannot be indifferent between the two realizable 
	fractions in her domain.
	
	Assuming that the condition in line $4$ is triggered, no agent strictly 
	preferring a pure coalition can be in a mixed coalition in $\pi$.  Thereby, 
	$i$ needs to strictly prefer a mixed coalition and needs to be part of a 
	pure coalition in $\pi$. However, from this it follows that all agents of 
	her type in mixed coalitions also strictly prefer a mixed coalition. 
	Consequently, $\pi'$ needs to include more mixed pairs than $\pi$. However, 
	this cannot be achieved without making some of the agents of the other type 
	worse off, as all agents from the other type weakly preferring a mixed 
	coalition are already placed in a mixed coalition in $\pi$. 
	
	Assuming that the condition in line $7$ is triggered, $i$ needs to be in a 
	pure coalition and strictly prefers a mixed coalition. However, it is not 
	possible that such an agent exists in this case, as this would imply that 
	either $b^{*}_{\min(k,\ell)}$ or $r^{*}_{\min(k,\ell)}$ also strictly 
	prefers a mixed coalition.  
	
	Assuming that the condition in line $9$ is triggered, it either holds that 
	$i$ strictly prefers a mixed coalition---and the reasoning from the first 
	case applies---or $i$ is member of $\{r^{'}_{q},b^{'}_{p}\}$ and strictly 
	prefers a pure coalition. In the latter case, every agent of $i$'s type who 
	is in a pure coalition  strictly prefers a pure coalition. Thereby, in 
	every outcome $\pi'$ that weakly Pareto dominates $\pi$, the number of pure 
	coalitions needs to be higher than in $\pi$. It is only possible 
	to increase the number of pure coalitions by putting at least one red and 
	one blue agent from $\{\{b^{*}_{i},r^{*}_{i}\} : i\in [\min(\ell,k)]\}$ in 
	a pure coalition. However, as the condition in line $7$ does not hold, at 
	least one of these agents strictly prefers a mixed coalition and is thereby 
	worse off in $\pi'$.
\end{proof}

%%%%%%%%%%%%%%%%%%%%%%%%%%%%%%%%%%%%%%%%%%%%%%%%%%%%%%%%%%%%%%

\section{(Strong) Core Stability}\label{sec:core}
We have seen that for the roommate diversity problem with $s=2$,
a core stable outcome is guaranteed to exist. 
However, for larger values of $s$ this is not the case.
\begin{theorem} \label{th:exCore}
	An instance of the roommate diversity problem
	may fail to have a core stable outcome, 
	even if no indifferences in the preferences are allowed. 
\end{theorem}
\begin{proof}
	Let $G=(\{r_{1},r_2,r_3,r_{4}\},\{b_{1},b_2,b_3,b_{4}\},4,
	(\succsim_{i})_{i\in N})$ 
	with
	\begin{align*}
	r_{1},r_{2},r_{3}:\frac{2}{4}\succ \frac{4}{4} \succ& \frac{3}{4} 
	\succ \frac{1}{4} \succ \frac04;
	\text{ } r_{4}:\frac{4}{4}\succ \frac{1}{4} \succ \frac{2}{4} 
	\succ \frac{3}{4} \succ \frac04\\
	b_{1}, b_{2}, b_{3},b_{4}&: \frac{1}{4}\succ \frac{2}{4}\succ \frac{3}{4}\succ \frac04 \succ \frac44. 
	\end{align*}
	Assume for the sake of contradiction that $G$ has a core stable outcome 
	$\pi$. If $\pi$ consists of two coalitions of fraction $\frac{2}{4}$, 
	$\{r_{4}, b_{1},b_{2},b_{3}\}$ is blocking. If $\pi$ consists of one 
	coalition of fraction $\frac{3}{4}$ and one coalition of fraction 
	$\frac{1}{4}$, $\{r_{1},r_{2},r_{3},r_{4}\}$ is blocking. If $\pi$ consists 
	of a purely blue and a purely red coalition, $\{r_{1},r_{2},b_{1},b_{2}\}$ 
	is blocking.
\end{proof}

\noindent Further, if we assume that the room size is an unfixed part of the 
input, the 
associated existence question becomes NP-complete.

\begin{theorem} \label{th::Roo_Core_NPc}
	It is {\em NP}-complete to decide whether a given instance of the roommate diversity 
	problem admits a core stable outcome. The hardness result holds even 
	if no indifferences in the preferences are allowed. 
\end{theorem}
\begin{proof}
    To check whether an outcome is core stable, we iterate over all $j\in[0,s]$ 
    and check if there are $j$ red agents and $s-j$ blue agents preferring 
    $\frac{j}{s}$ to the fraction of their current coalition.
	To prove hardness, we reduce from the problem 
	of deciding whether an anonymous hedonic game admits a core stable outcome 
	\cite{ballester2004np}. Given an anonymous game with $n$ agents, 
	we build an $n$-dimensional instance of the roommate diversity problem. 
	The general idea of the reduction is to map coalitions of size $\ell$ in 
	the 
	anonymous game to coalitions of fraction $\frac{\ell}{n}$ in the roommate 
	diversity problem.

	\textbf{Construction: }Let $G=(N, (\succ_{i})_{i\in N})$ be an anonymous 
	game. 
	In the corresponding roommate diversity problem, the room size $s$ is set 
	to 
	$|N|$, and for each original agent $i\in N$, one red agent $r_{i}$  is 
	introduced. The preference relation of $r_{i}$ is defined by: $$ 
	\frac{j}{s}\succsim_{r_{i}} \frac{j'}{s} \text{ iff } j\succsim_{i} j', 
	\forall 
	j,j'\in [s].$$ Additionally, for each $i\in [s],$ we insert $s^2$ blue 
	agents: 
	$$b_{i}^{j}: \frac{i}{s}\succ_{b_{i}^{j}} 0, \forall j\in [s^2].$$ 
	
	\textbf{Correctness: }($\Rightarrow$) Given a core stable outcome 
	$\pi=\{P_{1},\dots,P_{q}\}$ of the anonymous game, we construct a core 
	stable 
	outcome $\pi'$ of the corresponding roommate diversity problem. For each 
	$P_{\ell}$ with $\ell\in [q]$, we introduce a coalition $P'_{\ell}$ 
	consisting 
	of the red agents corresponding to the agents in $P_{\ell}$ together with 
	$s-|P_{\ell}|$ blue agents of type $b_{|P_{\ell}|}$: $P_{\ell}'=\{r_{i} : 
	i\in 
	P_{\ell}\}\cup \{b_{|P_{\ell}|}^{j} : j\in 
	[s(\ell-1)+1,s\ell-|P_{\ell}|]\}$. 
	All remaining blue agents are put into pure coalitions of size $s$. 
	
	For the sake of contradiction, let us assume that there exists a blocking 
	coalition $S$ with $|S|=s$ in $\pi'$. First of all, note that at least one 
	red 
	agent needs to be in $S$. Let us assume 
	that $S$ includes $j$ red agents $S_{r}\subseteq S$. For all $r_{i}\in 
	S_{r}$, 
	it holds that $\frac{j}{s}\succ_{r_{i}}\theta(\pi'(r_{i}))$. However, by 
	construction of the preference relation of $r_{i}$ and the outcome $\pi'$, 
	this 
	implies that for all $i\in \{\ell : r_{\ell}\in S_{r}\}$ it holds that 
	$j\succ_{i}\theta(\pi(i))$. Consequently, $\{\ell : r_{\ell}\in S_{r}\}$ 
	blocks 
	$\pi$ which leads to a contradiction. 
	
	$(\Leftarrow)$ To prove the other direction, let us assume that the 
	roommate 
	diversity problem admits a core stable outcome 
	$\pi=\{P_{1},\dots,P_{\ell},$ 
	$\dots,P_{k}\}$ where w.l.o.g. we assume that the first $\ell$ coalitions 
	contain red agents and all other coalitions are purely blue coalitions. 
	From 
	$\pi$, we construct an outcome of the anonymous game $\pi'=\{P'_{1},\dots, 
	P'_{\ell}\}$ with $P'_{j}=\{i : r_{i}\in P_{j} \}$ for all $j\in [\ell]$. 
	For 
	the sake of contradiction, let us assume that there exists a blocking 
	coalition 
	$S$ with $|S|=j$ in $\pi'$ for some $j\in [n]$. This implies that 
	$j\succ_{i}|\pi'(i)|$ for all $i\in S$. Let $S'=\{r_{i} : i\in S\}$. 
	By 
	construction, it holds that $S'$ is non-empty and that 
	$\frac{j}{s}\succ_{r_{i}}\theta(\pi(r_{i}))$ for all $r_{i}\in S'$. 
	Moreover, 
	note that there always exists a set $S_{b}$ of $s-j$ blue agents $b$ with $ 
	\frac{j}{s}\succ_{b} \theta(\pi(b))$, as at most $s*(s-1)$ blue agents can 
	be 
	part of a mixed coalition and $s^2$ blue agents prefer $\frac{j}{s}$ to a 
	pure 
	coalition. Therefore, $S'\cup S_{b}$ with $|S'\cup S_{b}|=s$ blocks $\pi$ 
	which 
	leads to a contradiction.
\end{proof}

Using the construction in the proof of Theorem~\ref{th::Roo_Core_NPc},
we can map the single-peaked anonymous hedonic game with empty core 
constructed by \citet{banerjee2001core} to a single-peaked instance 
of the roommate diversity problem with an empty core. 
We obtain the following corollary.
\begin{corollary} \label{th:exCoreSP}
	An instance of the roommate diversity problem may fail to have a core stable outcome, 
	even if all agents' preferences are single-peaked. 
\end{corollary}

In contrast, if agents' preferences are dichotomous, the core is non-empty, 
and an outcome in the core can be computed efficiently: following 
the approach of \citet{peters2016complexity}, to construct a core stable outcome 
we iterate over all fractions $\frac{\ell}{s}$ for $\ell \in[0,s]$ and add 
the maximum possible number of rooms 
consisting of $\ell$ red agents and $s-\ell$ blue agents who all approve 
$\frac{\ell}{s}$. The rest of the agents are split into 
the remaining rooms. We obtain the following result.
\begin{theorem} \label{th::RoomDich}
	Every instance of the roommate diversity problem with dichotomous preferences 
	admits a core stable outcome; moreover, an outcome in the core can be computed 
	in polynomial time.
\end{theorem}
\noindent However, this positive result does not extend to the more demanding notion of strong 
core stability.

\begin{theorem} \label{th::Roo_SC_Npc}
	It is {\em NP}-complete to decide whether a given
	roommate diversity problem 
	admits a strongly core stable outcome, even if the preferences are 
	restricted to be dichotomous and every agent approves at most four fractions. 
\end{theorem}
\begin{proof}
	\citet{peters2016complexity} showed that the corresponding problem for anonymous 
	hedonic games is NP-complete. By reducing from this problem, 
	we can establish that our problem is NP-hard as well; the reduction is similar
	to the one used in the proof of Theorem~\ref{th::Roo_Core_NPc}. 
\end{proof}

%%%%%%%%%%%%%%%%%%%%%%%%%%%%%%%%%%%%%%%%%%%%%%%%%%%%%%%%%%%%%%%%%%%%%%%%%%%%

\section{(Strong) Exchange Stability}
As pointed out in the introduction, it is not always plausible to assume that agents are 
allowed to perform group deviations. Therefore, in this section we focus on stability 
concepts that are defined in terms of agent swaps.

%%%%%%%%%%%%%%%%%%%%%%%%%%%%%%%%%%%%%%%%

\subsection{Same-Type Swaps} \label{se::STS}
If the set of possible deviations is limited to agent swaps, 
it may be the case that only agents of the same type are allowed 
to swap their places: we call the 
resulting stability notion {\em same-type-exchange stability}. For example, 
if a professor forms fixed-size teams for a group project, 
she may want to fix the fraction of graduate students in each group
(e.g., to ensure that the experienced students are equally distributed)
but still allow for swaps between two undergraduate students or 
between two graduate students. Under this weaker 
version of exchange stability, the agents are guaranteed
to eventually converge to a stable outcome.

\begin{theorem} \label{th::Room_SSwap}
	Every instance of the roommate diversity problem has a (strongly) 
	same-type-exchange stable outcome, and some such outcome  
	can be computed in polynomial time. 
\end{theorem} 
\begin{proof}
	To compute a (strongly) same-type-exchange stable outcome, 
	we start at an arbitrary outcome $\pi$ and swap pairs who have a 
	(weak) same-type-exchange deviation 
	until this is no longer possible. To see that this procedure always 
	terminates, note that same-type exchange deviations do not change the fraction 
	of red agents in any coalition. Thereby, nothing changes for agents 
	who are not involved in the swap, 
	while both agents involved in the swap 
	weakly improve and at least one of them strictly improves. 
	As every agent's preference relation is defined over $s+1$ elements, 
	the total number of swaps is at most $ns$.
\end{proof}

%%%%%%%%%%%%%%%%%%%%%%%%%%%%%%%%%%%%%%%%%%%%%%%%%%%%%%%%%%%%%%%

\subsection{Unrestricted Swaps}
Unfortunately, the existence guarantee for same-type swaps
does not extend to unrestricted swaps.

\begin{theorem} \label{th:SSSP}
	An instance of the roommate diversity problem may fail to have an exchange stable outcome,
	even if the preferences are single-peaked and 
	no indifferences in the preferences are allowed. 
\end{theorem}
\begin{proof}
	The following single-peaked roommate diversity game  
	$G=(\{r_{1},r_{2},r_{3},r_{4},r_{5}\},$\\ $\{b_{1},b_{2}, b_{3}, b_{4}\},
	3,(\succsim_{i})_{i\in N})$ with:
	\begin{alignat*}{2}
	r_{1},r_{2} & : \frac33\succ \frac{2}{3} \succ \frac{1}{3} \succ \frac03; \text{ }
	r_{3},r_{4},r_{5} && : \frac{2}{3} \succ \frac{1}{3} \succ \frac33 \succ \frac03; \\
	b_{1} & : \frac03 \succ \frac{1}{3} \succ \frac{2}{3} \succ \frac33;  \text{ }
	b_{2}, b_{3}, b_{4} && : \frac{1}{3}\succ \frac{2}{3} \succ \frac03 \succ\frac33.
	\end{alignat*}
	does not admit an exchange-stable outcome.
	
	We will first argue that there is no exchange stable outcome
	containing a purely red coalition $S$. Indeed, if such an outcome exists, 
	either (i) both of the remaining red agents are in the same coalition, or
	(ii) the remaining red agents are in two different coalitions. 
	
	In case (i), there also exists a purely blue coalition, which we denote by 
	$T$.
	At least one of $r_{3},r_{4}$ and $r_{5}$ has to be in $S$, and
	at least one of $b_{2},b_{3}$ and $b_{4}$ has to be in $T$;
	thus, at least one agent in $S$ and at least one agent in $T$
	would like to swap places.
	
	In case (ii), as $b_{1}$ strictly prefers $0$ to $\frac{1}{3}$, she 
	would like to swap with the red agent 
	in the other mixed coalition. As all red agents prefer $\frac{2}{3}$ to 
	$\frac{1}{3}$, this red agent would be happy to swap with $b_1$, too.
	
	It remains to argue that there is no exchange stable outcome with no purely 
	red coalition.
	In this case, there exist two mixed coalitions of fraction $\frac{2}{3}$ 
	and one mixed coalition of fraction $\frac{1}{3}$. Consequently, at least 
	one of 
	$r_{1}$ and $r_{2}$ belongs to a coalition of fraction $\frac{2}{3}$;
	this agent would like to swap with the blue agent in the 
	other mixed coalition of fraction $\frac{2}{3}$. As all blue agents prefer 
	$\frac{1}{3}$ to $\frac{2}{3}$, this blue agent 
	would be happy with the swap as well, so no such outcome
	can be exchange stable.
\end{proof}

\noindent Further, the associated existence problem for strong exchange stability is NP-complete.
\begin{theorem} \label{th:SSS}
	It is {\em NP}-complete to decide whether a given instance of 
	the roommate diversity problem admits 
	a strongly exchange stable outcome. 
	The hardness result still holds if
	the preferences are dichotomous and every agent approves at most four fractions. 
\end{theorem}

\begin{proof}
	To show membership, note that it is possible to verify in polynomial time 
	whether an outcome is strongly exchange stable by iterating over all pairs 
	of agents and checking whether they have a weak exchange-deviation.

	To show hardness, we reduce from the problem of deciding 
	whether an anonymous hedonic game admits a Nash stable 
	outcome~\citep{ballester2004np};
	recall that an outcome of a hedonic game is said to be {\em Nash stable}
	if no agent wants to move from her current coalition to another existing 
	coalition
	or to form a singleton coalition. Our reduction is similar to the 
	one in the proof of Theorem~\ref{th::Roo_Core_NPc};
	however, here, we introduce $n^2-n$ blue agents who are indifferent among 
	all fractions. Thereby, each time an agent has a Nash deviation in the 
	anonymous game,
	the corresponding red agent has a weak exchange-deviation with a blue 
	agent 
	in the
	corresponding outcome of the constructed roommate diversity problem and the 
	other way round.

	\textbf{Construction: }Given an anonymous game $G=(N, (\succ_{i})_{i\in 
	N})$, the room size $s$ in the corresponding roommate diversity problem is 
	set to $|N|$. Moreover, for each agent $i\in N$, a red agent $r_{i}$ is 
	introduced whose preference relation is defined by: $$ 
	\frac{j}{s}\succsim_{r_{i}} \frac{j'}{s} \text{ iff } j\succsim_{i} j', 
	\forall j,j'\in [s].$$ Further, for all $i\in [s^{2}-s]$, a blue agent 
	$b_{i}$ who is indifferent among all fractions is introduced. 
	
	\textbf{Correctness: }$(\Rightarrow)$ We start by assuming that $G$ admits 
	a Nash stable outcome $\pi=\{P_{1},\dots,P_{q}\}$. From $\pi$, we construct 
	a strongly exchange stable outcome of the corresponding roommate diversity 
	problem. First of all, for all $\ell \in [q]$, we insert one coalition 
	$P'_{\ell}$ consisting of all red agents corresponding to the agents in 
	$P_{\ell}$ and $s-|P_{\ell}|$ blue agents into $\pi'$: $P'_{\ell}=\{r_{i} : 
	i\in 
	P_{\ell}\}\cup \{b_{i} : i\in [s(\ell-1)+1,s\ell-|P_{\ell}|]\}$. The 
	remaining blue 
	agents are put into pure coalitions of size $s$. $\pi'$ is not guaranteed 
	to be strongly exchange stable, as there might exist a pair of red agents 
	with an exchange-deviation. However, as $\pi$ is Nash stable, for all $i\in 
	[n]$ it holds that:
	$$\theta(\pi'(r_{i}))\succsim_{r_{i}}\frac{1}{s} \text{ and }\forall j\in 
	[q]\setminus\{\pi'(r_{i}) 
	\}:\theta(\pi'(r_{i}))\succsim_{r_{i}}\frac{|P_{j}|+1}{s}.$$
	To construct a strongly exchange stable outcome from $\pi'$, we execute 
	weak same-type-exchange-deviations as long as there exists a pair with a 
	weak 
	same-type-exchange-deviation to obtain an outcome $\pi''$.\footnote{As 
	discussed in 
	Theorem~\ref{th::Room_SSwap}, this process is always guaranteed to 
	terminate.} It remains to 
	show that no two agents of a different type have a weak exchange-deviation 
	in 
	$\pi''$: recall that by performing weak same-type-exchange-deviations every 
	agent's situation can only be improved. Consequently, in $\pi''$, it holds 
	for all $i\in [n]$ that: 
	$\theta(\pi''(r_{i}))\succsim_{r_{i}}\theta(\pi'(r_{i}))$. Applying this to 
	the equation from above, it follows that no red agent has an incentive to 
	swap with a blue agent in $\pi''$. Consequently, $\pi''$ is strongly 
	exchange stable. 
	
	($\Leftarrow$) To prove the other direction, let us assume that the 
	roommate diversity problem admits a strongly exchange stable outcome 
	$\pi=\{P_{1},\dots,P_{\ell}, \dots, P_{k}\}$ where  w.l.o.g. the first 
	$\ell$ coalitions contain red agents, while the remaining coalitions are 
	purely blue coalitions. Let $\pi'=\{P^{'}_{1},\dots,P^{'}_{\ell}\} $ with 
	$P^{'}_{j}=\{i : r_{i}\in P_{j}\}$ for all $j\in [\ell]$ be an outcome of 
	the anonymous game. To prove that $\pi'$ is Nash stable, note that in 
	$\pi$, for all $i\in N$, there always exists at least one blue agent in 
	every coalition $r$ is not part of. As $\pi$ is strongly exchange stable, 
	from this and the fact that each blue agent is indifferent among all 
	coalitions it follows that for all $i\in [n]$ it holds that: $$
	\theta(\pi(r_{i}))\succsim_{r_{i}}\frac{1}{s} \text{ and }\forall j\in 
	[\ell]\setminus\{\pi(r_{i})\}:\theta(\pi(r_{i}))\succsim_{r_{i}}\frac{|P'_{j}|+1}{s}.
	$$ By construction, from this it follows that for all $i\in [n]$ it holds 
	that: $$\theta(\pi'(i))\succsim_{i} 1 \text{ and }  \forall j\in 
	[\ell]\setminus\{\pi'(i)\}:\theta(\pi'(i))\succsim_{i}|P'_{j}|+1,$$ which 
	implies that $\pi'$ is Nash stable. 
	
	It is also possible to show that the proven NP-hardness result still holds 
	if the agents' preferences are dichotomous and every agent approves at most 
	four fractions by reducing from the related problem for anonymous games, 
	which is NP-complete as proven by \citet{peters2016complexity}. The 
	construction of the reduction needs to be slightly adapted: for all red 
	agents, the set of approved fractions is the transformed set of the 
	corresponding agent in the anonymous game, and all blue agents do not 
	approve any fraction.  
\end{proof}

We believe that deciding whether a roommate diversity problem admits an exchange stable outcome
is NP-complete as well, but we were unable to extend the proof of Theorem~\ref{th:SSS} 
to show this.

%%%%%%%%%%%%%%%%%%%%%%%%%%%%%%%%%%%%%%%%%%%%%%%%%%%%%%%%%%%%%%%%%%%%%%%%%%%%

\section{Pareto Optimality}
In the roommate problem, Pareto optimality emerges as a natural notion of 
stability. Indeed, an outcome is not Pareto optimal if and only if we 
can rearrange the agents so that all of them are weakly better off and at least some
of them are strictly better off, i.e., there is a weakly improving 
deviation by the grand coalition \citep{morrill2010roommates,elkind2016price}. 

While for many other stability concepts we consider 
stable outcomes are not guaranteed to exist, by the definition of Pareto optimality, every instance of the roommate diversity
problem admits a Pareto optimal outcome. Indeed, we can start at an arbitrary outcome
and perform a sequence of at most $ns$ Pareto improvements, i.e., rearrangements of the agents
that make all agents weakly better off and some agents strictly better off. 
However, it is still computationally hard to compute a Pareto optimal outcome, as finding a Pareto improvement is difficult.

\begin{theorem}\label{th:Rea}
	For the unrestricted roommate diversity problem, 
        it is {\em coNP}-complete to decide whether a given outcome 
	is Pareto optimal; 
	moreover, we cannot compute a Pareto optimal outcome 
	in polynomial time unless {\em P=NP}. These results hold even if preferences 
        are dichotomous.
\end{theorem}
\begin{proof}
	To show that an outcome is not Pareto optimal, it suffices to guess a Pareto
	improvement and verify that it indeed makes all agents weakly better off and some
	agents strictly better off; this establishes that our decision problem is 
	in coNP.
	
	For both hardness results, we utilize the fact that \citet{aziz2013pareto} 
	proved that deciding whether an outcome of an anonymous game is Pareto 
	optimal is coNP-complete and computing a Pareto optimal outcome is NP-hard. 
	Note that their proofs also hold for the case where agents' preferences are 
	dichotomous and every agent approves at 
	most four sizes. 
	
	To show hardness of verifying whether a given outcome is Pareto optimal, we 
	construct a reduction from the related problem for anonymous 
	games. The reduction works analogously to the reduction in the proof of 
	Theorem~\ref{th:SSS}. To prove the correctness of the reduction, let us 
	assume that 
	$\pi=\{P_{1},\dots,P_{q}\}$ is a Pareto optimal outcome of the original 
	anonymous game. Let  $\pi'=\{P'_{1},\dots,P'_{k}\}$ be the outcome of the 
	corresponding roommate diversity problem where all agents are replaced by 
	their corresponding red agent and free spots are filled with blue agents. 
	Then, $\pi'$ needs to be Pareto optimal, as every outcome that weakly 
	Pareto dominates $\pi'$ can be converted into an outcome that weakly Pareto 
	dominates $\pi$ by removing all blue agents and replacing all red agents by 
	their corresponding agent in the anonymous game. This reasoning also holds 
	for the other direction.
	
	Secondly, computing a Pareto optimal outcome in a roommate diversity 
	problem is NP-hard: it is possible to apply the reduction from 
	Theorem~\ref{th:SSS} to transform every anonymous game into a corresponding 
	roommate 
	diversity problem. Assuming that we have found a Pareto optimal outcome of 
	this roommate diversity problem, it is possible to obtain a Pareto optimal 
	outcome of the original anonymous game as described above. 
\end{proof}

%%%%%%%%%%%%%%%%%%%%%%%%%%%%%%%%%%%%%%%%%%%%%%%%%%%%%%%%%%%%%%%%%%%%%%%%%%%%%%%%%%%%%%%%%
\section{Envy-Freeness}
We can think of envy-freeness as a ``one-sided'' version of exchange stability. 
Thus, similarly to same-type-exchange stability, we can define same-type-envy-freeness
by only considering envy among agents of the same type. 
Same-type-envy-freeness is a plausible variant of envy-freeness, 
as people tend to envy those who are similar to them \citep{salovey1984some}.
Moreover, same-type-envy-freeness 
is also an appealing notion of fairness: if agent $i$ and agent $j$
are of the same type and $i$ envies $j$, swapping $i$ and $j$ has no effect
on other agents, so the decision which of these agents should get a better set
of roommates is essentially arbitrary.
Unfortunately, an outcome that is fair in this sense is not guaranteed to exist. 

\begin{theorem} \label{th::Room_Envy}
	There exists an instance of the roommate diversity problem with room size two 
	that has no same-type-envy-free outcome and thereby also no envy-free outcome.
	Moreover, in this instance the agents' preferences are 
	single-peaked and dichotomous.
\end{theorem}
\begin{proof}
	Let $G=(\{r_{1}\},\{b_{1},b_{2},b_{3}\},2,(\succsim_{i})_{i\in N})$ 
	with: 
	$$
	r_{1}: \frac22 \succ \frac{1}{2} \succ \frac02;\qquad b_{1},b_{2},b_{3}: 
	\frac02\succ \frac{1}{2} \succ \frac22.
	$$
	This game is clearly single-peaked and can be transformed into a 
	dichotomous game where all agents disapprove their two bottom fractions 
	and approve their top-fraction. 
	As every outcome of this game consists of one mixed and one purely blue coalition, 
	the blue agent in the mixed coalition always envies the two blue agents in the pure 
	coalition.
\end{proof}

On the positive side, there exist two special cases where the existence 
of a same-type-envy-free outcome is guaranteed.

\begin{theorem}
	A same-type-envy-free outcome is guaranteed to exist if the number of red agents 
	is divisible by $s$ or by $k$.
\end{theorem}
\begin{proof}
	If $s$ divides $r$, there exists an outcome consisting of pure coalitions only.
	If $k$ divides $r$, there exists an outcome
	where the fraction of each coalition is $\frac{r}{k}$. In either case, 
	all agents of the same type are 
	in coalitions of the same fraction, so no agent envies another agent of her type.
\end{proof}

\noindent Nevertheless, it can be proven by a reduction from \textsc{Exact Cover by 3-Sets} that the general existence question for envy-freeness is NP-complete.

\begin{theorem} \label{th::Room_Envy_NPc}
	It is {\em NP}-complete to decide whether a given instance of 
	the roommate diversity problem admits an envy-free outcome, even 
	if no indifferences in the preference relations are allowed. 
	This hardness result also holds if 
	the agents' preferences are dichotomous 
	and every agent is allowed to approve 
	at most four fractions.
\end{theorem}
\begin{proof}
	To prove membership in NP, note that it is possible  to check in polynomial 
	time whether an outcome is envy-free by iterating over all pairs of agents 
	and checking whether one of them envies the other. 
	
	In the following, we utilize the fact that the problem {\scshape Exact 
	Cover By 3-Sets} 
	({\sc X3C}) is NP-complete \cite{garey2002computers}. An instance of this 
	problem is given 
	by a set $X=\{1,\dots,m\}$ and a collection $C=\{A_{1}, \dots, A_{q}\}$ 
	of 3-element subsets of $X$. It is a yes-instance if there exists a subset 
	$C'\subseteq C$ such that $C'$ is a partition of $X$. 
	
	To prove hardness, we construct a reduction from {\sc X3C}. Let 
	$X=\{1,\dots,m\}$ and $C=\{A_{1}, \dots, A_{q}\}$ be a collection of 
	3-element subsets of $X$. For $i\in X$, let 
	$J^{i}=\{j^{i}_{1},\dots,j^{i}_{m_{i}}\}$ be the set of all indices of sets 
	in $C$ to which $i$ belongs, i.e. $ j\in J^{i} $ if and only if $i\in 
	A_{j}$.  In the following, we construct a corresponding roommate diversity 
	problem. The general idea of the construction is to introduce for each 
	element $i\in X$ a corresponding agent $r_{i}$ such that in an envy-free 
	outcome for all $A_{j}\in C$ there either exists a coalition of fraction 
	$\frac{5j+1}{s}$ including the three agents corresponding to the elements 
	in $A_{j}$ or a coalition of fraction $\frac{5j-2}{s}$. 
	
	\textbf{Construction: }To begin with, let us set $s=5q+1$ and let us 
	introduce one red so-called set agent $r_{i}$ for each $i\in X$ with the 
	following preference relation:  
	$$r_{i}: \frac{5j_{1}^{i}+1}{s}\sim_{r_{i}} \dots \sim_{r_{i}} 
	\frac{5j_{m_{i}}^{i}+1}{s}\succ_{r_{i}}\frac{1}{s}\succ_{r_{i}}\dots \text{ 
	}.$$
	Moreover, for each $j\in [q]$, we introduce $5j-2$ red so-called redundant 
	agents  with the following preference relation: 
	$$r_{j}^{p}: \frac{5j+1}{s}\sim_{r_{j}^{p}} 
	\frac{5j-2}{s}\succ_{r_{j}^{p}}\frac{1}{s}\succ_{r_{j}^{p}}\dots, \text{ } 
	\forall p\in [5j-2].$$
	Furthermore, for each $j\in [q]$, we insert $s-(5j+1)$ blue so-called 
	filling agents with the following preference relation: 
	$$b_{j}^{p}: \frac{5j+1}{s}\sim_{b_{j}^{p}} 
	\frac{5j-2}{s}\succ_{b_{j}^{p}}0 \succ_{b_{j}^{p}}\dots , \text{ } \forall 
	p\in [s-(5j+1)].$$
	Additionally, for each $j\in [q],$ we introduce three blue so-called 
	additional agents $\tilde{b}^{1}_{j},\tilde{b}^{2}_{j},\tilde{b}^{3}_{j}$ 
	with the following preference relation: 
	$$\tilde{b}^{1}_{j},\tilde{b}^{2}_{j},\tilde{b}^{3}_{j}: 
	\frac{5j-2}{s}\succ 0 \succ \dots \text{ }.$$
	We further insert blue agents strictly preferring $0$ to all other 
	fractions until the total number of agents is divisible by $s$.    
	Lastly, for each red agent of the game, we introduce $s$ blue agents 
	strictly preferring $0$ to all other fractions.
	
	\textbf{Correctness: }$(\Rightarrow)$ We start by assuming that there 
	exists a partition $\pi \subseteq C$ of $X$. In the following, using $\pi$, 
	we construct an envy-free outcome $\pi'$ of the roommate diversity problem. 
	Therefore, for each $A_{j}\in \pi$, we add one coalition $P_{j}$ consisting 
	of the corresponding three red set agents, the $5j-2$ designated red 
	redundant agents and the $s-(5j+1)$ designated blue filling agents: 
	$P_{j}=\{r_{i} : i\in A_{j}\}\cup \{r_{j}^p : p\in [5j-2] \}\cup \{b_{j}^p 
	: p\in [s-(5j+1)] \} $. Moreover, for each $A_{j}\notin \pi$, we create a 
	coalition $P_{j}$ consisting of the designated $5j-2$ red redundant agents 
	together with the designated $s-(5j+1)$ blue filling agents and the 
	designated three additional blue agents: 
	$P_{j}=\{\tilde{b}^{1}_{j},\tilde{b}^{2}_{j},\tilde{b}^{3}_{j}\}\cup 
	\{r_{j}^p : p\in [5j-2] \}\cup \{b_{j}^p : p\in [s-(5j+1)] \}$. Finally, we 
	put the remaining blue agents into purely blue coalitions. $\pi'$ consists 
	of purely blue coalitions and $q$ mixed coalitions of fraction 
	$\frac{5j+1}{s}$ or $\frac{5j-2}{s}$ for $j\in [q]$. 
	
	We prove that $\pi'$ is envy-free by iterating over all agents and showing 
	that they do not envy another agent: as $\pi$ is a partition, all red set 
	agents are in one of their most preferred coalitions. Moreover, by 
	construction, all redundant agents and all filling agents  are in one of 
	their most preferred coalitions.

	Concerning the additional agents, for all $j\in [q]$ such that $A_{j}\notin 
	\pi$, $\tilde{b}^{1}_{j},\tilde{b}^{2}_{j},\tilde{b}^{3}_{j}$ are in their 
	top coalition. For all $j\in [q]$ such that $A_{j}\in \pi$, 
	$\tilde{b}^{1}_{j},\tilde{b}^{2}_{j},\tilde{b}^{3}_{j}$ are in a purely 
	blue coalition. However, as $A_{j}\in \pi$, there does not exist a 
	coalition of fraction $\frac{5j-2}{s}$. As there also does not exist a 
	coalition within distance $\frac{1}{s}$ to this fraction, 
	$\tilde{b}^{1}_{j},\tilde{b}^{2}_{j}$ and $\tilde{b}^{3}_{j}$ do not envy 
	another agent for all $j\in [q]$.
	
	As all blue agents strictly preferring a purely blue coalition to all other 
	coalitions are in a pure coalition, no agent envies another agent in 
	$\pi'$. 
	
	$(\Leftarrow)$ Let us assume that there exists an envy-free outcome $\pi$ 
	of the constructed roommate diversity problem. In the following, we prove 
	that this implies that there exists a partition $\pi'\subseteq C$ of $X$. 
	First of all, note that due to the introduction of $s$ blue agents for each 
	red agent, every outcome of the game includes at least one purely blue 
	coalition $P\in \pi$. 
	
	Using this, we prove that each red agent is part of one of her most 
	preferred coalitions. For the sake of contradiction, let us assume that 
	there exists a red agent $p$ for which this is not the case. Then, it 
	either holds that $p$ is in a coalition of fraction $\frac{1}{s}$ or in a 
	coalition of a fraction to which she prefers $\frac{1}{s}$. In the former 
	case, all blue agents in $i$'s coalition envy the agents in $P$, as all 
	blue agents prefer $0$ to $\frac{1}{s}$. In the latter case, $p$ envies the 
	agents in $P$. Consequently, every red agent is in one of her most 
	preferred coalitions.
	
	From this it follows that every coalition of fraction $\frac{5j+1}{s}$ for 
	some $j \in [q]$ includes at most $5j-2$ red redundant agents  and thereby 
	exactly three set agents $r_{i}$ with $i\in A_{j}$. By removing all non-set 
	agents from $\pi$, it is possible to obtain a partition $\pi'\subseteq C$ 
	of $X$.
	
	Note that the reduction still holds if no indifferences in the agents' 
	preferences are allowed. It is possible to replace the indifferences in the 
	agents' preference relations by strict preferences in an arbitrary way 
	without influencing the validity of the proof. The second part of the proof 
	remains  unaffected, while the outcome $\pi'$ constructed in the first part 
	is still envy-free, as for all $j \in [q]$ only either a coalition of 
	fraction $\frac{5j-2}{s}$ or of fraction $\frac{5j+1}{s}$ exists and for 
	all $j,j'\in [q]$ it holds that $|\frac{5j+1}{s}-\frac{5j'-2}{s}|> 
	\frac{1}{s}$. 
	
	Note further that it is also possible to transform every {\sc X3C} instance
	into a corresponding dichotomous roommate diversity problem. To do so, the 
	construction given above needs to be only slightly adapted: for all agents, 
	their set of approved fractions consists of all explicitly listed 
	fractions.  
\end{proof}

We were not able to extend
Theorem~\ref{th::Room_Envy_NPc}
to same-type-envy-freeness, but conjecture that the hardness result still holds. 
If preferences are strict, there exists a simple algorithm 
that solves this problem in time linear in $n$ and single-exponential in $s$, 
i.e., this problem is in FPT with respect to $s$.

\begin{theorem}\label{th::EF}
	Given an instance of 
	the roommate diversity problem with room size $s$ and strict preferences,
	it is possible to check in time $\mathcal{O}(ns2^s)$ whether 
	this instance admits a same-type-envy-free outcome
	and to find one if it exists.
\end{theorem} 
\begin{proof}
	Let $\pi=\{C_1, \dots, C_k\}$ be a same-type-envy-free outcome. 
	Then, for every agent $i$ there does not exist a coalition $C\in\pi$
	such that $i$ strictly prefers $\theta(C)$
	to the fraction of her own coalition $\theta(\pi(i))$.
	
	This observation gives rise to the following algorithm. 
	First, we guess a subset 
	$X\subseteq \{\frac{\ell}{s} : \ell\in [0,s] \}$ of precisely the fractions 
	that occur in one same-type-envy-free outcome of the game. 
	Then, for each $\frac{\ell}{s} \in X$, we determine the set $S_\ell$ of all agents 
	for whom $\frac{\ell}{s}$ is the most preferred element of $X$. 
	If for each $\frac{\ell}{s}\in X$, the set
	$S_\ell$ includes $t\cdot \ell$ red agents and $t\cdot (s-\ell)$ 
	blue agents for some positive integer $t$, 
	we create $t$ rooms of fraction $\frac{\ell}{s}$ consisting of exactly the agents from $S_\ell$; if for some fraction in $X$ this is not the case,
	we reject the current guess of $X$. If all guesses are rejected, there is no same-type-envy-free outcome.
\end{proof}

%%%%%%%%%%%%%%%%%%%%%%%%%%%%%%%%%%%%%%%%%%%%%%%%%%%%%%%%%%%%%%%%%%%%%%%%%%%%%%%%%%
\section{Parameterized Analysis}\label{sec:param}
In Section \ref{se::Two}, we saw that fixing the size of the rooms to $s=2$ has a 
significant impact on the complexity of finding stable outcomes.
Motivated by this result, as well as by Theorem~\ref{th::EF}, 
in this section, we study the parameterized complexity of the roommate 
diversity problem
with respect to parameter $s$. This is a promising parameter, 
since in most of our hardness reductions we converted
an anonymous hedonic game with $n$ agents into an instance with room size $n$; 
it is also appealing because in practice
the room size can be much smaller than the number of agents. 
Indeed, most of the algorithmic problems 
considered in this work turn out to be in FPT with respect to $s$.
We start by considering (strong) core stability.

\begin{theorem} \label{th::core_FPT}
	The problem of determining whether an instance of the roommate diversity problem 
	admits a (strongly) core stable 
	outcome is fixed-parameter tractable with respect to the room size. 
\end{theorem}
\begin{proof}
Throughout the proof, we focus on finding core stable outcomes; in the end, 
we will comment on how to modify the proof for strong core stability.

As every agent is fully characterized by her preference relation and type, 
both the number of different red agents $t_{r}$ and 
the number of different blue agents $t_{b}$ 
can be upper-bounded by $2^{(s+1)^2}$. 
Let $\succsim^{R}_{1},...,\succsim^{R}_{t_{r}}$ 
(respectively, $\succsim^{B}_{1},...,\succsim^{B}_{t_{b}}$) 
be an enumeration of all possible preference relations of red 
(respectively, blue) agents. Then, an instance of our 
problem can be fully described by a $(t_{r}+t_{b})$-tuple 
$(r_{1},...,r_{t_{r}},b_{1},...,b_{t_{b}})$, where 
$r_{i}$ (resp., $b_{i}$) denotes the number of red (resp., blue) 
agents with preference relation $\succsim^{R}_{i}$ 
(resp., $\succsim^{B}_{i}$).

Now, we need a concise encoding of outcomes that enables 
us to check whether a given outcome is core stable. 
For a given outcome, 
let $r_{i,j}$ (resp., $b_{i,j}$) denote the number of red 
(resp., blue) agents with preference relation $\succsim^{R}_{i}$ 
(resp., $\succsim^{B}_{i}$) in rooms with $j$ red agents. 
Moreover, given $r_{i,j}$ and $b_{i,j}$, let $q^{R}(j)$ 
(resp., $q^{B}(j)$) denote 
the number of red (resp., blue) agents who strictly prefer 
a room with $j$ red agents to their current room:	
$$
q^{R}(j)=\sum_{\ell\in [s]}\hspace*{0.1cm}
\sum_{\substack{i\in [t_{r}]: \\ \frac{j}{s}\succ_{i}^{R} 
		\frac{\ell}{s}}} r_{i,\ell}, \qquad
q^{B}(j)=\sum_{\ell\in [0,s-1]}\hspace*{0.1cm}
\sum_{\substack{i\in [t_{b}]: \\ \frac{j}{s}\succ_{i}^{B} 
		\frac{\ell}{s}}} b_{i,\ell}.
$$ 
The values of these variables determine whether an outcome is core stable: 
for example, there exists a blocking coalition with three red agents 
if and only if 
$q^{R}(3)\geq 3$ and $q^{B}(3)\geq s-3$.

Hence, to decide whether an instance 
$G=(r_{1},...,r_{t_{r}},b_{1},...,b_{t_{b}})$ admits a 
core stable outcome, it is sufficient to check whether there exists an 
assignment to the 
variables $r_{i,j}$ and $b_{i,j}$ inducing an outcome for which no blocking 
coalition exists. We formalize this problem as an Integer Linear Program (ILP) 
and use the 
fact that it is possible to solve an ILP with $\rho$ variables and input length 
$L$ in time 
$\mathcal{O}(\rho^{2.5\rho+o(\rho)}L)$ 
\cite{lenstra1983integer,kannan1987minkowski}. 
To build the ILP, for each $j\in [0,s]$, we define an
additional variable $n_{j}$ denoting the number of rooms with $j$ red agents. 
The intuition detailed above results in the following collection of constraints:
\begin{align*}
&\sum_{j\in [s]} r_{i,j}  =r_{i}, \forall i\in [t_{R}]&   
&\sum_{j\in [0,s-1]} b_{i,j}  =b_{i}, \forall i\in [t_{B}]   \tag{1}\\[1em]
&\sum_{i\in [t_{R}]} r_{i,j} = jn_{j},  \forall j\in [s]&   
&\sum_{i\in [t_{B}]} b_{i,j} = (s-j)n_{j},  \forall j\in [0,s-1] \tag{2}
\\
&q^{R}(j) < j\vee q^{B}(j) < (s-j),&   &\forall j\in[0,s] \tag{3}\\[0.5em]
&b_{i,j},r_{i',j}, n_{j} \geq 0,&   
&\forall j \in [0,s], \forall i \in [t_{B}], \forall i' \in [t_{R}] \tag{4}
\end{align*}

First, constraints~(1) ensure that every agent is assigned to some 
coalition. Second, constraints~(2) ensure that the values of $r_{i,j}$ and 
$b_{i,j}$ 
induce a valid outcome, i.e., it is possible to put all agents into rooms of 
fraction 
specified by the variables $r_{i,j}$ and $b_{i,j}$. Constraints~(3) ensure that 
the induced 
outcome $\pi$ is core stable by enforcing that there is no blocking coalition 
of 
size $s$ with $j$ red and $s-j$ blue agents, for all $j\in [0,s]$; these 
constraints
are non-linear but can be converted to linear constraints by introducing
one extra variable per constraint \cite[Chapter~9]{bradley1977applied}.

The resulting ILP has $\mathcal{O}(s2^{s^2})$ variables and
has a feasible solution if and only if the 
underlying roommate diversity problem admits a core stable outcome. 
Now, applying the 
algorithm of \citet{lenstra1983integer} and \citet{kannan1987minkowski}, we can 
determine whether a roommate diversity problem admits a core stable outcome 
in time $f(s)\mathcal{O}(n)$ for some computable function $f(s)$.

By modifying the definitions of $q^R(j), q^B(j)$ and constraints (3) 
appropriately, 
we can adapt the ILP to check whether a roommate diversity problem admits a 
strongly core stable 
outcome.
\end{proof}

This approach, with appropriate modifications, extends to (strong) exchange stability and 
envy-freeness.

\begin{theorem} \label{th::swap_FPT} 
	The problem of determining whether an instance of the roommate diversity problem 
	admits a (strongly) exchange stable outcome  
	is fixed-parameter tractable with respect to the room size. 
\end{theorem}
\begin{proof}
	To construct an algorithm for this problem, we build upon the ideas and 
	notation from Theorem \ref{th::core_FPT}. However, it is not possible to 
	directly apply the reasoning from the previous theorem to this situation, 
	as here it is necessary to bookmark the fraction of the agents' current 
	coalition as well. Therefore, let $c^{R}(j,k)$ (resp., $c^{B}(j,k)$) denote 
	the number of red (resp., blue) agents being in a coalition with $j$ red 
	agents and preferring a coalition with $k$ red agents:
	$$c^{R}(j,k)=\sum_{\substack{i\in [t_{r}]: \\ \frac{k}{s}\succ_{i}  ^{R} 
	\frac{j}{s}}} r_{i,j},\qquad 
	c^{B}(j,k)=\sum_{\substack{i\in [t_{b}]: \\ \frac{k}{s}\succ_{i}  ^{B} 
	\frac{j}{s}}} b_{i,j}.$$
	
	To begin with, we only consider different-type swaps. To ensure that an 
	outcome is different-type-exchange stable, it is sufficient to check that 
	there does not exist a red agent $r^*\in R$ and a blue agent $b^*\in B$ in 
	different coalitions satisfying the following property: there exists a 
	$j\in[s]$ and $k\in[s]$  such that $r^*$ is in a coalition with $j$ red 
	agents and prefers a coalition with $k$ red agents, and $b^*$ is in a 
	coalition with $k-1$ red agents and prefers a coalition with $j-1$ red 
	agents. If two such agents exist, they always have a 
	different-type-exchange-deviation. In the case that $j\neq k-1$, two agents 
	satisfying the specified property are always in different coalitions. 
	Consequently, every different-type-exchange stable outcome has to satisfy 
	$c^{R}(j,k)=0\lor c^{B}(k-1,j-1)=0$ for all $j,k\in [s]: j\neq k-1 $.  In 
	contrast to this, in the case that $j=k-1$, it might be possible to put all 
	agents satisfying the specified property in the same coalition. However, if 
	this is not possible, the outcome becomes immediately unstable. Therefore, 
	every different-type-exchange stable outcome has to additionally satisfy 
	$c^{R}(j,j+1)=0 
	\vee c^{B}(j,j-1)=0 \vee  (c^{R}(j,j+1)\leq j \wedge c^{B}(j,j-1)\leq s-j)$ 
	for all $j\in [s-1].$ 
	
	On the other hand, to ensure that $\pi$ is same-type-exchange stable, it is 
	sufficient to check that for no $k\neq j\in [0,s]$ do there exist two 
	agents of the same type such that one of them is in a coalition with $j$ 
	red agents and prefers a coalition with $k$ red agents, and the other is in 
	a coalition with $k$ red agents and prefers a coalition with $j$ red 
	agents. Note that if two such agents exist they are in different coalitions 
	and thereby always have an exchange-deviation. Consequently, every 
	same-type exchange stable outcome has to satisfy $c^{R}(j,k)=0 \lor 
	c^{R}(k,j)=0$ for all $j,k\in [s]:j<k$ and  $c^{B}(j,k)=0 \lor 
	c^{B}(k,j)=0$ for all $j,k\in [0,s-1]:j< k$. 
	
	Putting all these ideas together results in the following collection of 
	constraints: 
	\allowdisplaybreaks
	\begin{align*}
	&\sum_{j\in [0,s-1]} b_{i,j}  =b_{i}, \forall i\in [t_{B}];   &   
	&\sum_{j\in [s]} r_{i,j}  =r_{i}, \forall i\in [t_{R}] \tag{1}\\
	&\sum_{i\in [t_{B}]} b_{i,j} = (s-j)n_{j}, \forall j\in [0,s-1];       &   
	&\sum_{i\in [t_{R}]} r_{i,j} = jn_{j}, \forall j\in [s] \tag{2}\\ 
	&c^{R}(j,k)=0\lor c^{B}(k-1,j-1)=0,&   &\forall j,k\in [s]: j\neq k-1 
	\tag{3.a}\\
	&c^{R}(j,j+1)=0 \vee c^{B}(j,j-1)=0 \hspace*{0.1cm}\vee&   &\\ 
	&(c^{R}(j,j+1)\leq j \wedge c^{B}(j,j-1)\leq s-j),&   &\forall j\in [s-1]. 
	\tag{3.b}\\ 
	&c^{R}(j,k)\leq 0 \vee c^{R}(k,j)\leq 0,&       &   \forall j,k\in [s]: j<k 
	\tag{4}\\
	&c^{B}(j,k)\leq 0\vee c^{B}(k,j)\leq 0,&       &   \forall j,k\in [0,s-1]: 
	j<k \tag{5}\\
	& b_{i,j},r_{i',j}, n_{j} \geq 0 , &   &\forall j \in [0,s], \forall i \in 
	[t_{B}], \forall i' 
	\in [t_{R}] \tag{6}
	\end{align*}
	As in the proof of Theorem~\ref{th::core_FPT}, the first two constraints 
	ensure that 
	$r_{i,j}$ and $b_{i,j}$ induce a valid outcome $\pi$, while constraints 
	(3.a) and (3.b) ensure that no pair of agents with a 
	different-type-exchange-deviation exists. Moreover, constraints (4) and (5) 
	ensure that no pair of agents with a same-type-exchange-deviation exists. 
	Again, it is possible to apply the ideas of 
	\citet[Chapter~9]{bradley1977applied} to obtain an ILP in standard form 
	where the number of used variables lies in $\mathcal{O}(s2^{s^2})$. 
	Therefore, by applying the algorithm from \citet{lenstra1983integer} and 
	\citet{kannan1987minkowski}, it is possible to determine whether a roommate 
	diversity problem admits an exchange stable outcome in $f(s)\mathcal{O}(n)$ 
	for some computable function $f(s)$. 
	
	It is also possible to adapt the definitions of $c^R(j,k), c^B(j,k)$ and 
	constraints (3)-(5) slightly to check whether a strongly exchange stable 
	outcome exists.  
\end{proof}
\begin{theorem} \label{th::envy_FPT}
	The problem of determining whether an instance of the roommate diversity problem 
	admits an envy-free outcome 
	is fixed-parameter tractable with respect to the room size. 
\end{theorem}
\begin{proof}
	First of all, recall that envy-freeness can be interpreted as a one-sided 
	version of exchange stability where only one agent needs to be willing to 
	swap with the other. Therefore, we only need to slightly adapt the proof of 
	Theorem~\ref{th::swap_FPT} to show the theorem. Again, we use the notation 
	introduced in 
	the previous two proofs. We only discuss the constraints that red agents 
	need to fulfill in an envy-free outcome. The constraints for blue agents 
	follow analogously.

	We start by considering same-type-envy. In a same-type-envy-free outcome, 
	for no $j,k\in [1,s]$, there exists a red agent $r^*$ who is in a coalition 
	with $j$ red agents and prefers an existing coalition with $k$ red agents, 
	i.e. 
	$\frac{k}{s}\succ_{r^*} \frac{j}{s}$. This would immediately imply that 
	$r^*$ envies another red agent in one of the coalitions with $k$ agents. 
	Thereby, it needs to hold that $c^{R}(j,k)= 0 \lor n_k= 0$ for all $j,k\in 
	[s]$.

	In a different-type-envy-free outcome, for no red agent $r^*$ in a 
	coalition with $j$ red agents who prefers a coalition with $k$ agents, a 
	different coalition with $k-1$ red agents is allowed to exist. In the case 
	where $j\neq k-1$, if one coalition with $k-1$ exists, $r^*$ cannot be part 
	of this coalition, as $r^*$ is part of a coalition with $j$ red agents. 
	Thereby, if there exists at least on coalition with $k-1$ red agents, 
	$r^*$ always envies at least one blue agent in such a coalition in this 
	case. However, if $j=k-1$ and $n_j=1$, $r^*$ is in the only coalition with 
	$k-1$ red agents and thereby there exists no blue agent in a different 
	coalition with $k-1$ red agents who $r^*$ could envy. These conditions can 
	be enforced by the following two constraints: $c^{R}(j,k)=0\lor n_{k-1}=0$ 
	for all $j,k\in [s]: j\neq k-1$ and $c^{R}(j,j+1)=0 \lor n_{j}\leq 1$ for 
	all $j\in [s-1]$.
	
	These considerations result in the following collection of constraints: 
	\begin{align*}
	&\sum_{j\in [0,s-1]} b_{i,j}  =b_{i}, \forall i\in [t_{B}];   &   
	&\sum_{j\in [s]} r_{i,j}  =r_{i}, \forall i\in [t_{R}] \tag{1}\\
	&\sum_{i\in [t_{B}]} b_{i,j} = (s-j)n_{j}, \forall j\in [0,s-1];       &   
	&\sum_{i\in [t_{R}]} r_{i,j} = jn_{j}, \forall j\in [s] \tag{2}\\ 
	&c^{R}(j,k)= 0 \lor n_k= 0,&   &\forall j,k\in [s] \tag{3}\\
	&c^{R}(j,k)=0\lor n_{k-1}=0,&   &\forall j,k\in [s]: j\neq k-1 \tag{4.a}\\ 
	&c^{R}(j,j+1)=0 \lor n_{j}\leq 1&   &\forall j\in [s-1]\tag{4.b}\\ 
	&c^{B}(j,k)= 0 \lor n_k= 0,&   &\forall j,k\in [0,s-1] \tag{5}\\
	&c^{B}(j,k)=0\lor n_{k+1}=0,&   &\forall j,k\in [0,s-1]: j\neq k+1 
	\tag{6.a}\\ 
	&c^{B}(j,j-1)=0 \lor n_{j}\leq 1&   & \forall j\in [s-1]\tag{6.b}\\ 
	& b_{i,j},r_{i',j}, n_{j} \geq 0 , &   &\forall j \in [0,s], \forall i \in 
	[t_{B}], \forall i' 
	\in [t_{R}] \tag{7}
	\end{align*}
	As in the proof of Theorem~\ref{th::core_FPT}, the first two constraints 
	ensure that 
	$r_{i,j}$ and $b_{i,j}$ induce a valid outcome $\pi$, while constraint (3) 
	ensures that no red agent envies another red agent. Moreover, constraints 
	(4.a) and (4.b) ensure that no red agent envies a blue agent. 
	Respectively, constraints (5), (6.a) and (6.b) ensure that no blue agent 
	envies another blue agent or a red agent. 
	Again, it is possible to apply the ideas of 
	\citet[Chapter~9]{bradley1977applied} to obtain an ILP in standard form 
	where the number of used variables lies in $\mathcal{O}(s2^{s^2})$. 
	Therefore, by applying the algorithm from \citet{lenstra1983integer} and 
	\citet{kannan1987minkowski}, it is possible to determine whether a roommate 
	diversity problem admits an envy-free outcome in $f(s)\mathcal{O}(n)$ for 
	some computable function $f(s)$. 
\end{proof}

Theorems~\ref{th::core_FPT}--\ref{th::envy_FPT} highlight a fundamental difference 
between the classical stable roommate problem and the roommate problem with 
diversity preferences: 
for the former all studied existence problems are NP-complete for $s= 3$ or even for $s=2$, 
while for the latter many of these problems are polynomial-time solvable if the size of the rooms 
is fixed or bounded. Thus, assuming diversity preferences fundamentally 
changes the complexity of the roommate problem. We note, however, that we were unable to
extend our ILP techniques to questions concerning Pareto optimality; it remains an open
problem whether these questions are also in FPT with respect to the room size $s$.

\section{Extensions Of The Model}\label{sec:ext}
One can generalize the roommate diversity problem, as defined in this paper, 
in several ways. This section proposes three independent extensions;
exploring them in detail is left for future work. 

\smallskip

\noindent \textbf{\textit{More than two types}\ } 
\citet{be20} generalized hedonic diversity games to allow for more than two 
types.
They proposed two different approaches. Under the first approach, 
agents are allowed to express their preferences over every 
possible combination of fractions of agents of all other types in their 
coalition: 
for example, assuming that red, blue and green agents are present, an agent is 
allowed to assess 
a coalition based on the fraction of red agents, the fraction of blue agents 
and the 
fraction of green agents in the coalition. Under the second approach, 
each agent only cares about the fraction of agents of her own type in her 
coalition:
e.g., a red agent is indifferent between a coalition with 2 red agents and 2 
blue agents and 
a coalition with 2 red agents, 1 blue agent and 1 green agent.
We can apply both of these approaches in our setting. Notably, the first 
approach 
results in a formalism that subsumes the multidimensional stable roommate 
problem 
with general preferences, as it is possible to create a designated type for 
each agent. 
Consequently, all negative results for the stable roommate problem hold for 
this model. 
However, the hardness results require an unbounded number of types. 
In fact, all existence questions considered in this paper still lie in 
$\mathit{FPT}$ with respect to the size of the rooms plus the number of types:
we can introduce a variable for the number of agents of a given type with a 
given preference
relation that are placed in a room with a given distribution of types 
and adjust the constraints in the ILPs accordingly. 
%and therefore become polynomial-time solvable if we fix both the number of 
%types $t$ and the 
%size of the rooms $s$. 
%Indeed, instead of fixing the number of red agents in every room, 
%we can to fix the number of agents of each type in every room 
%($\mathcal{O}(s^t)$), and to adjust the constraints in the ILPs accordingly.

\smallskip

\noindent \textbf{\textit{Free spots }\ } 
In our model, we assume that agents fit exactly into the available rooms. In 
practice, this 
is not always the case. We can extend our model by adding a 
parameter $k$ specifying the number of rooms, so that $ks\geq |N|$. Since it is 
no longer 
guaranteed that each room has exactly $s$ agents, the domain of the agents' 
preferences 
needs to be extended accordingly: $D=\{\frac{j}{t}: t\in [s], j\in [0,t]\}$.

In this model, we can consider modifications of traditional hedonic game 
stability notions, 
such as Nash stability, individual stability and core stability, 
where agents (groups) are only allowed to deviate into free spots; 
this approach has been explored by \citet{aziz2013stable} for the classic 
roommate setting.
In fact, this model is expressive enough to represent 
every hedonic diversity game: it suffices to set $s=k=n$.
Thus, all negative results for hedonic diversity games such as the hardness of 
determining 
whether a game admits a Nash or core stable outcome 
\citep{bredereck2019hedonic,be20} 
also hold in this model.

\smallskip

\noindent \textbf{\textit{Non-uniform room sizes }} Another natural extension  
is to relax the uniformity condition imposed on the size of the rooms 
and allow rooms of different sizes. For instance, in our  
wedding reception example, different tables may have different sizes. 
To capture this, we replace the constant $s$ in our model by a tuple 
$(s_1,....,s_k)$, 
where $s_i$ denotes the size of the $i$-th room; 
we require $\sum_{i=1}^{k}s_i=|N|$.
Again, it is necessary to adjust the 
domain of the agents' preferences: $D=\{\frac{j}{s_i}: i\in [k], j\in 
[0,s_i]\}$. 
Then, all existence problems considered in Section~\ref{sec:param}
can be shown to be fixed-parameter tractable with respect to 
$\max_{i\in [k]}s_i$, by a simple modification of the respective ILPs.

%%%%%%%%%%%%%%%%%%%%%%%%%%%%%%%%%%%%%%%%%%%%%%%%%%%%%%%%%%%%%%%%%%%%%%%%%%%%%%%%%%%%%
\section{Marriage Problem With Diversity Preferences}
Recall that in the stable marriage problem, we have $n$ men and $n$ women,
who have preferences over the members of the opposite sex,  
and the goal is to split them into $n$ pairs, with each pair consisting of 
one man and one woman, so that the resulting assignment is core stable;
one can also consider other notions of fairness and stability in this context, 
such as exchange stability or Pareto optimality. 
%Most of the 
%computational questions associated with this model are polynomial-time
%solvable \cite{irving1994stable}, at least when agents have strict preferences
%(see, however, the work of \citet{cechlarova2005exchange} on exchange 
%stability). 

\citet{knuth1976mariages}
generalized this model to higher dimensions, where the agents belong to $s\ge 2$
categories (say, freshmen, sophomores, juniors and seniors), and the goal is to 
form
groups of size $s$, with one representative from each category. Here, 
preferences are either defined 
over $(s-1)$-tuples of agents or induced by preferences over each category. 
For preferences defined over tuples, \citet{ng1991three} and, independently, 
\citet{subramanian1994new} proved NP-completeness of determining whether a game 
admits a 
core stable outcome; for the setting where preferences are defined 
over individual agents and then lifted to tuples,  
\citet{huang2007two} proved NP-completeness of several stability-related 
questions.

In this section, we consider a variant of this problem with diversity 
preferences: 
we assume that each agent belongs to one of $s$
categories and also to one of two types (say, red and blue) . The goal is to 
form groups
that have one representative from each category, with agents' preferences over 
groups
being determined by the fraction of red agents in each group. For 
instance, 
a project team may have to include one student from each year of study 
(i.e., a freshman, a sophomore, a junior and a senior), but the students'
preferences over teams are driven by the gender distribution in each team
(assuming for simplicity that each student identifies as male or female).
%As another example, for an interdisciplinary research project, the relevant
%catogories may be researchers' subject areas (e.g., economics, mathematics, 
%computer science), while researchers' preferences may be defined 
%in terms of experience with specific problems (such as, e.g., hedonic games).
%
We formalize this setting as follows.

\begin{definition}
	An $s$-dimensional {\em marriage diversity problem} is a quadruple 
	$G=(\{N_1,...,N_s\},$ $R,B, (\succsim_{i})_{i\in N_1\cup\dots\cup N_s})$ 
	with $N=\bigcupdott_{i\in [s]}N_i$, $N=R\cupdott B$, and $|N|=k\cdot s$ 
	for some $k\in \mathbb{N}$. The preference relation $\succsim_{i}$ 
	of each agent $i\in N$ is a weak order over the set 
	$D=\{\frac{j}{s} : j\in [0,s]\}$.
\end{definition}

An outcome of a marriage diversity problem is a partition 
$\pi=\{C_1,...,C_k\}$ of all agents into $k$ different coalitions of size $s$ 
such that 
$|C_i\cap N_j|=1$ for all $i\in [k]$, $j\in [s]$.
In what follows, we refer to $N_i$ as the {\em $i$-th dimension} 
of the problem and say that two agents 
$a,b \in N$ {\em lie in the same dimension} if $a, b\in N_i$
for some $i\in [s]$. The definitions of stability concepts for this setting
can be obtained by modifying the respective definitions for the roommate
diversity problem: for exchange stability, we only allow swaps between agents
that lie in the same dimension, and for core stability, we only allow deviations
by coalitions that contain exactly one agent from each dimension.

Interestingly, in contrast to the roommate diversity problem, 
the marriage diversity problem may fail to have an exchange stable outcome even 
for $s=2$. 

\begin{theorem}
	An instance of the marriage diversity problem may fail to have an exchange 
	stable outcome,
	even if $s=2$ and agents' preferences are dichotomous and single-peaked.
\end{theorem}
\begin{proof}
	Consider an instance of the marriage diversity problem with 
	$R=\{r_1, r_2\}$, $B=\{b_1, b_2\}$, $s=2$, $N_1=\{r_1, b_1\}$, 
	$N_2=\{r_2, b_2\}$ and the following preferences:
	$$
	r_{1}: \frac22\succ \frac{1}{2} \succ \frac02,\hspace*{0.25cm} 
	r_{2}: \frac{1}{2}\succ \frac22 \succ \frac02,\hspace*{0.25cm} 
	b_{1}: \frac02\succ \frac{1}{2} \succ \frac22, \hspace*{0.25cm}
	b_{2}: \frac{1}{2}\succ \frac02 \succ \frac22.
	$$
	This instance has two possible outcomes: 
	$\pi=\{\{r_1,r_2\},\{b_1,b_2\}\}$ and $\pi'=\{\{r_1,b_2\},\{b_1,r_2\}\}$.
	In $\pi$,  agents $r_2$ and $b_2$ have an exchange deviation, and 
	in $\pi'$, agents $r_1$ and $b_1$ have an exchange deviation. 
	This example also shows that an envy-free outcome is not guaranteed to 
	exist.
\end{proof}

While the marriage diversity problem is more structured than the roommate 
diversity problem, 
the two problems are strongly related:
we can adapt the reductions from Theorem~\ref{th::Roo_Core_NPc} and 
Theorem~\ref{th:SSS} to prove that hardness results for core and strong 
exchange 
stability still hold for the marriage diversity problem.

\begin{theorem} \label{th::MaEx}
	It is {\em NP}-complete to decide whether a given instance of the marriage 
	diversity problem 
	admits a core stable outcome or a strongly exchange stable outcome. 
\end{theorem}
\begin{proof}
	To begin with, we prove the theorem for core stability by constructing a 
	reduction from the related problem for anonymous games:
	
	\textbf{Construction: }Let $G=(N, (\succ_{i})_{i\in N})$ be an anonymous 
	game. We construct a corresponding $|N|$-dimensional marriage diversity 
	problem. For each $i\in N$, we create a new dimension containing one red 
	agent $r_{i}$ whose preference relation is defined by: $$ 
	\frac{j}{s}\succsim_{r_{i}} \frac{j'}{s} \text{ iff } j\succsim_{i} j', 
	\forall j,j'\in [s]$$ and, for each $\ell\in [s],$ $s^2$ blue agents: 
	$$b_{\ell}^{j}: \frac{\ell}{s}\succ_{b_{\ell}^{j}} 0, \forall j\in [s^2].$$
	
	The proof of correctness of this reduction is completely analogous to the 
	proof of Theorem~\ref{th::Roo_Core_NPc}.
	
	\smallskip
	\noindent Next, we prove hardness for strong exchange stability by 
	constructing a 
	reduction from the problem of deciding whether an anonymous game admits a 
	Nash stable outcome:
	
	\textbf{Construction: } Let $G=(N, (\succ_{i})_{i\in N})$ be an anonymous 
	game. From $G$, we create an $|N|$-dimensional marriage diversity problem. 
	For each $i\in N$, we create a new dimension containing a red agent $r_{i}$ 
	whose preference relation is defined by: $$ \frac{j}{s}\succsim_{r_{i}} 
	\frac{j'}{s} \text{ iff } j\succsim_{i} j', \forall j,j'\in [s]$$ and $n-1$ 
	blue agents who are indifferent among all coalitions.
	
	\textbf{Correctness: }$(\Rightarrow)$ Let  $\pi=\{P_{1},\dots,P_{q}\}$ be a 
	Nash stable outcome of the anonymous game. From~$\pi$, we construct a 
	strongly exchange stable outcome of the corresponding marriage diversity 
	problem. For all $j \in [q]$, we insert one coalition $P'_{j}$ consisting 
	of all red agents corresponding to the agents in~$P_{j}$ and $s-|P_{j}|$ 
	blue agents from the remaining dimensions into $\pi'$. The remaining blue 
	agents are put into valid pure coalitions of size $s$.
	
	$\pi'$ is strongly exchange stable, as no pair of red agents is allowed to 
	swap, as they are all from different dimensions. Furthermore, no red agent 
	has an incentive to swap with a blue agent, as this would correspond to an 
	individual deviation of the corresponding agent in the anonymous game which 
	cannot be rational, since $\pi$ is Nash stable. 
	
	($\Leftarrow$) Let $\pi=\{P_{1},\dots,P_{\ell}, \dots, P_{k}\}$ be a 
	strongly exchange stable outcome of the marriage diversity problem where  
	w.l.o.g. the first $\ell$ coalitions contain red agents, while the 
	remaining coalitions are purely blue coalitions. Let 
	$\pi'=\{P^{'}_{1},\dots,P^{'}_{\ell}\} $ with $P^{'}_{j}=\{i : r_{i}\in 
	P_{j}\}$ for all $j\in [\ell]$ be an outcome of the anonymous game. To 
	prove that $\pi'$ is Nash stable, note that for all red agents there always 
	exists exactly one blue agent from her dimension in every other coalition. 
	As $\pi$ is strongly exchange stable, from this and the fact that each blue 
	agent is indifferent among all coalitions it follows that for all~$i\in 
	[n]$ it holds that: $$
	\theta(\pi(r_{i}))\succsim_{r_{i}}\frac{1}{s} \text{ and }\forall j\in 
	[\ell]\setminus\{\pi(r_{i})\}:\theta(\pi(r_{i}))\succsim_{r_{i}}\frac{|P'_{j}|+1}{s}.
	$$ By construction, from this it follows that for all $i\in [n]$ it holds 
	that: $$\theta(\pi'(i))\succsim_{i} 1 \text{ and }  \forall j\in 
	[\ell]\setminus\{\pi'(i)\}:\theta(\pi'(i))\succsim_{i}|P'_{j}|+1,$$ which 
	implies that $\pi'$ is Nash stable. 
\end{proof}

Also, just as for the roommate problem, we can transform the single-peaked 
anonymous
hedonic game with empty core constructed by~\citet{banerjee2001core} 
into an instance of the marriage diversity problem with no core stable outcome.
\begin{corollary}
	A marriage diversity problem may fail to admit a core stable outcome, even 
	if preferences are strict and single-peaked.
\end{corollary}

On the positive side, we can apply our ILP approach to the marriage 
diversity problem to show that the problem of finding a (strongly) core stable 
or a (strongly) exchange stable outcome is fixed-parameter tractable with 
respect to
the number of dimensions.
\begin{theorem} \label{th::MaFPT}
	The problem of deciding whether a marriage diversity problem admits a 
	(strongly)
	core stable or a (strongly) exchange stable outcome is in $\mathit{FPT}$ 
	with respect to the number of dimensions $s$.
\end{theorem}
\begin{proof}
	To begin with, we adapt the ILP from Theorem~\ref{th::core_FPT} to prove 
	the theorem for 
	core stability. Here, apart from bookmarking the type and preference 
	relation of every agent, it is also necessary to bookmark her 
	dimension. Therefore, given a marriage diversity problem, let 
	$r_{i}^{\ell}$ (resp.,~$b_{i}^{\ell}$) be the number of red (resp., blue) 
	agents with preference relation $\succsim^R_i$ (resp., $\succsim^B_i$) 
	which are part of the $\ell$-th dimension and $r_{i,j}^{\ell}$ (resp., 
	$b_{i,j}^{\ell}$) the number of such agents in a room with $j$ red agents. 
	
	A blocking coalition containing $j$ red agents exists in an outcome $\pi$ 
	if there exist $j$ dimensions where at least one red agent prefers such a 
	coalition to her current coalition and a disjoint set of $s-j$ dimensions 
	where at least one blue agent prefers such a coalition to her current 
	coalition. To check this, for each $j\in[s]$, let $g_{j}$ denote the number 
	of dimensions where at least one red agent prefers a coalition of $j$ red 
	agents to her current coalition and at least one blue agent a coalition of 
	$j$ red agents to her current coalition. Moreover, let $g_{j}^R$ denote the 
	number of dimensions where at least one red agent prefers a coalition of 
	$j$ red agents to her current coalition and no blue agent prefers a 
	coalition of $j$ red agents to her current coalition. We define $g_{j}^B$ 
	analogously. Utilizing these variables, it is possible to ensure that no 
	blocking coalition exists by enforcing that for each $j$ the gap between 
	$j$ and~$g_{j}^R$ plus the gap between $s-j$ and $g_{j}^B$ is less than 
	$g_{j}$. This idea results in the following collection of adapted 
	constraints: 
	
	\begin{align*}
	&\sum_{j\in [s]} r^{\ell}_{i,j}  =r^{\ell}_{i}, \forall i\in [t_{R}]&   
	&\sum_{j\in [0,s-1]} b^{\ell}_{i,j}  =b^{\ell}_{i}, \forall i\in [t_{B}]   
	\tag{1}\\[1em]
	&\sum_{\substack{i\in [t_{R}]\\\ell\in[s]}} r^\ell_{i,j} = jn_{j},  \forall 
	j\in [s]&   &\sum_{\substack{i\in [t_{B}]\\\ell\in[s]}} b^\ell_{i,j} = 
	(s-j)n_{j},  \forall j\in [0,s-1] \tag{2.a}
	\\
	&\sum_{i\in [t_{R}]} r^\ell_{i,j}+ \sum_{i\in [t_{B}]} b^\ell_{i,j}= 
	n_{j},&   &\forall j\in [0,s], \forall \ell\in [s] \tag{2.b}
	\\
	&\max(0,j-g_j^R)+ & &\\
	&\max(0,(n-j)-g_j^B) <g_j,&   &\forall j\in[0,s] \tag{3}\\[1em]
	&b^\ell_{i,j},r^\ell_{i',j}, n_{j} \geq 0,&   &\forall j \in [0,s], \forall 
	i \in 
	[t_{B}], \forall i'\in [t_{R}], \forall \ell\in[s] \tag{5}
	\end{align*}
	First, constraints~(1) ensure that every agent in the game is assigned to 
	some room. Second, constraints~(2.a) and~(2.b) ensure that the values of 
	$r^\ell_{i,j}$ and $b^\ell_{i,j}$ induce a valid outcome of the game: 
	constraints (2.a) ensure that it is possible to put all agents into rooms 
	of fraction specified in the variables $r^\ell_{i,j}$ and $b^\ell_{i,j}$, 
	while constraints (2.b) ensure that it is possible to construct such an 
	outcome where all agents in a coalition come from different dimension. As 
	explained above, constraints (3) ensure that no blocking coalition exists.
	Following again the ideas of \citet[Chapter~9]{bradley1977applied}, it is 
	possible to convert this into an ILP in standard form formulated in terms 
	of~$\mathcal{O}(s2^{s^2})$ variables. Applying the algorithm from 
	\citet{lenstra1983integer} and \citet{kannan1987minkowski}, it is possible 
	to determine whether a marriage diversity problem admits a core stable 
	outcome in $f(s)\mathcal{O}(n)$ for some computable function $f(s)$. 
	
	Turning to exchange stability, it is again possible to start from the 
	related Theorem~\ref{th::swap_FPT} for the roommate diversity problem. 
	However, in this 
	case, for two agents to have an exchange-deviation, they need to come from 
	the same dimension.  Therefore, let $c^{R}_\ell(j,k)$ (resp., 
	$c^{B}_\ell(j,k)$) be the number of red (resp., blue) agents from dimension 
	$\ell$ being in a coalition with $j$ red agents and preferring a coalition 
	with $k$ red agents. Here, in contrast to the roommate setting, we do not 
	need to check whether two agents which may have a exchange-deviation come 
	from two different coalitions, as this constraint is already implied by the 
	fact that the two agents need to be from the same dimension. Thereby, 
	adapting the ILP from Theorem~\ref{th::swap_FPT} results in the following 
	collection of 
	constraints: 
	\begin{align*}
	&\sum_{j\in [0,s-1]} b^{\ell}_{i,j}  =b^{\ell}_{i}, \forall i\in [t_{B}]&   
	&\sum_{j\in [s]} r^{\ell}_{i,j}  =r^{\ell}_{i}, \forall i\in 
	[t_{R}]\tag{1}\\
	&\sum_{\substack{i\in [t_{B}]\\\ell\in[s]}} b^\ell_{i,j} = (s-j)n_{j},  
	\forall j\in [0,s-1]&   &\sum_{\substack{i\in [t_{R}]\\\ell\in[s]}} 
	r^\ell_{i,j} = jn_{j},  \forall j\in [s] \tag{2.a}\\
	&\sum_{i\in [t_{R}]} r^\ell_{i,j}+ \sum_{i\in [t_{B}]} b^\ell_{i,j}= 
	n_{j},&   &\forall j\in [0,s],\forall \ell\in [s] \tag{2.b}\\
	&c^{R}_\ell(j,k)=0\lor c^{B}_\ell(k-1,j-1)=0,&   &\forall j,k,\ell\in [s] 
	\tag{3}\\
	&c^{R}_\ell(j,k)=0 \vee c^{R}_\ell(k,j)=0,&       &   \forall j,k,\ell\in 
	[s]: j<k \tag{4}\\
	&c^{B}_\ell(j,k)=0\vee c^{B}_\ell(k,j)=0,&       &   \forall 
	\ell\in[s], \forall j,k\in [0,s-1]: j<k \tag{5}\\
	&b_{i,j},r_{i',j}, n_{j} \geq 0,&   &\forall j \in [0,s], \forall i \in 
	[t_{B}], \forall i' \in [t_{R}] \tag{5}
	\end{align*}
	Here, as before, constraints (1), (2.a) and (2.b) ensure that 
	$r_{i,j}^\ell$ and $b_{i,j}^\ell$ induce a valid outcome. Moreover, 
	constraints (3) ensure that no two agents have a 
	different-type-exchange-deviation. In addition, constraints (4) ensure that 
	no two red agents have a same-type-exchange-deviation, while constraints 
	(5) do the same for blue agents. 
	It is again possible to convert this to an ILP in standard form following 
	the ideas of \citet[Chapter~9]{bradley1977applied} and to apply the 
	algorithm from \citet{lenstra1983integer} and \citet{kannan1987minkowski} 
	to determine whether a marriage diversity problem admits an exchange stable 
	outcome in $f(s)\mathcal{O}(n)$ for some computable function $f(s)$.
\end{proof}
%\textit{Proof idea:} To adjust the ILPs to these problems, it is necessary to 
%bookmark the 
%dimension of every agent and to adjust the constraints accordingly. Thereby, 
%for example, 
%checking for a blocking coalition, it is also necessary to check whether the 
%blocking 
%coalition includes one agent per dimension.

%%%%%%%%%%%%%%%%%%%%%%%%%%%%%%%%%%%%%%%%%%%%%%%%%%%%%%%%%%%%%%%%%%%%

\section{Conclusions And Future Directions}
In this paper, we have proposed the roommate diversity problem,
which is an interesting special case of the
multidimensional stable roommate problem. We have initiated the algorithmic study 
of this problem by considering various standard solution concepts and 
analyzing the existence of stable outcomes and the complexity of computing them 
(see Table~\ref{t:Sum}). While we have answered many questions that arise in this context, 
the complexity of deciding whether a roommate diversity problem 
admits an exchange stable outcome remains open. Moreover, it is unclear 
if every instance with dichotomous preferences admits 
an exchange stable outcome. 

By parameterizing our computational problems by the size of the rooms, we 
showed 
that diversity preferences are a powerful restriction, as all studied existence 
problems lie 
in FPT with respect to this parameter. However, 
it would be desirable to obtain parameterized algorithms that are combinatorial
rather than ILP-based, since ILP-based algorithms tend to be slow in practice.

Finally, we believe that both the marriage diversity problem and the three 
extensions of our model
introduced in Section~\ref{sec:ext} deserve further attention.

\section* {Acknowledgements}
This work was supported by 
a DFG project ``MaMu'', NI 369/19 (Boehmer) and by
an ERC Starting Grant ACCORD under Grant Agreement 639945 (Elkind).
\bibliographystyle{ACM-Reference-Format}
\bibliography{RDP_arxiv}

\end{document}